\documentclass{article}
\usepackage{multirow}

\usepackage{amsmath}
\usepackage{amsthm}
\usepackage{amssymb}
\usepackage{cases}
\usepackage{bm}

\usepackage{color}
\usepackage[usenames,dvipsnames]{xcolor}

\usepackage{tikz}
\usetikzlibrary{shapes,arrows}

\usepackage[style=numeric,citestyle=numeric,useprefix,backend=bibtex]{biblatex}
\bibliography{FLRWbib}

\newtheorem{thm}{Theorem}[section]  
\newtheorem{lem}[thm]{Lemma}
\newtheorem{prop}[thm]{Proposition}
\newtheorem{defn}[thm]{Definition}

\numberwithin{table}{section}
\numberwithin{figure}{section}

\allowdisplaybreaks 

\usepackage[colorlinks=true,linkcolor=blue,citecolor=red,urlcolor=Violet,bookmarksdepth=subsection,final]{hyperref}

  \newcommand{\del}{\partial}
  
  \newcommand{\eps}{\varepsilon}
  \newcommand{\CC}{\mathit{CC}}
  \newcommand{\ESU}{\mathit{ESU}}
  \newcommand{\CSC}{\mathit{CSC}}
  \newcommand{\FLRW}{\mathit{FLRW}}
  \newcommand{\CS}{\mathit{CS}}

  \newcommand{\CES}{\mathit{CES}}
  \newcommand{\MMS}{\mathit{MMS}}
  \newcommand{\NKG}{\mathit{NKG}}
  
\title{IDEAL characterization of isometry classes of FLRW and inflationary spacetimes}
\author{%
	Giovanni Canepa,$^{1,a}$
	Claudio Dappiaggi$^{2,3,b}$ and
	Igor Khavkine$^{4,c}$ \\[1ex]
	{\small $^1$ Universit\"at Z\"urich, Winterthurerstrasse 190, 8057 Z\"urich, Switzerland } \\
	{\small $^2$ Dipartimento di Fisica, Universit\`a degli Studi di Pavia,}\\ 
	\small{Via Bassi, 6, I-27100 Pavia, Italy} \\
	{\small $^3$ Istituto Nazionale di Fisica Nucleare -- Sezione di Pavia,} \\
	{\small Via Bassi, 6, I-27100 Pavia, Italy} \\
	{\small $^4$ Universit\`a di Milano, Via Cesare Saldini, 50, I-20133
	Milano, Italy} \\[1ex]
	{\small
		$^a$ \texttt{giovanni.canepa@math.uzh.ch} , \quad
		$^b$ \texttt{claudio.dappiaggi@unipv.it} ,} \\
	{\small
		$^c$ \texttt{igor.khavkine@unimi.it}
	}
}
\date{\today}

\begin{document}
\maketitle

\begin{abstract}
In general relativity, an IDEAL (Intrinsic, Deductive, Explicit,
ALgorithmic) characterization of a reference spacetime metric $g_0$
consists of a set of tensorial equations $T[g]=0$, constructed
covariantly out of the metric $g$, its Riemann curvature and their
derivatives, that are satisfied if and only if $g$ is locally isometric
to the reference spacetime metric $g_0$. The same notion can be extended
to also include scalar or tensor fields, where the equations
$T[g,\phi]=0$ are allowed to also depend on the extra fields $\phi$. We
give the first IDEAL characterization of cosmological FLRW spacetimes,
with and without a dynamical scalar (inflaton) field. We restrict our
attention to what we call regular geometries, which uniformly satisfy
certain identities or inequalities. They roughly split into the
following natural special cases: constant curvature spacetime, Einstein
static universe, and flat or curved spatial slices. We also briefly
comment on how the solution of this problem has implications, in general
relativity and inflation theory, for the construction of local gauge
invariant observables for linear cosmological perturbations and for
stability analysis.
\end{abstract}

\section{Introduction}

In this work, we are interested in an intrinsic characterization of
homogeneous and isotropic cosmological spacetimes (also known as
Friedmann-Lema\^itre-Ro\-bert\-son-Walker or \emph{FLRW spacetimes}),
either with or without the presence of a scalar field (aka
\emph{inflationary spacetimes}). By a spacetime $(M,g)$, we 
mean a smooth manifold $M$ with a Lorentzian metric $g$. While
``intrinsic'' generally does preclude direct reference to the form of
the spacetime metric in a special coordinate system, it is a vague
enough term to have multiple interpretations. To be specific, we
refer to an \emph{IDEAL}%
	\footnote{The acronym, explained in~\cite{Ferrando2010} (footnote, p.2),
	stands for Intrinsic, Deductive, Explicit and ALgorithmic.} %
or \emph{Rainich-type} characterization that has been used, for
instance, in the works~\cite{rainich, takeno, coll-ferrando,
Ferrando1998, Ferrando2009, Ferrando2010, krongos-torre}. It consists of
a list of tensorial equations ($T_a[g] = 0$, $a=1,2,\ldots,N$),
constructed covariantly out of the metric ($g$) and its derivatives
(concomitants of the Riemann tensor) that are satisfied if and only if
the given spacetime locally belongs to the desired class, possibly
narrow enough to be the isometry class of a reference spacetime
geometry. This notion has a natural generalization ($T_a[g,\phi] = 0$)
to spacetimes equipped with scalar or tensor fields ($\phi$), with
equivalence still given by isometric diffeomorphisms that also transform
the additional scalars or tensors into each other. A nice historical
survey of this and other local characterization results can be found
in~\cite{Mars2017}.

An IDEAL characterization neither requires the existence of any extra
geometric structures, nor the translation of the metric and of the
curvature into a frame formalism. Thus, it is an alternative to the
Cartan-Karlhede characterization~\cite[Ch.9]{stephani-sols}, which is
based on Cartan's moving frame formalism. Intrinsic characterizations, of
various types, have been of long standing and independent interest in
geometry and General Relativity. But, in addition, they can be helpful
in deciding when a metric, given for instance by some complicated
coordinate formulas, corresponds to one that is already known. In this
regard, an IDEAL characterization is especially helpful if one would
like to find an algorithmic solution to this recognition problem. In
numerical relativity, the near-satisfaction of the tensor equations
$T_a[g]\approx 0$ may signal the local proximity of a numerical
spacetime to a desired reference geometry. In addition, the approach to
zero of $T_a[g] \to 0$ could be used to study either linear or nonlinear
stability of reference geometries, in an unambiguous and gauge
independent way.

The following particular application should be noted. By the
Stewart-Walker lemma~\cite[Lem.2.2]{sw-pert}, the vanishing of a tensor
concomitant $T_a[g] = 0$ for a metric $g$ implies that its linearization
$\dot{T}_a[h]$ ($T_a[g+\eps h] = T_a[g] + \eps \dot{T}_a[g] +
O(\eps^2)$) is invariant under linearized diffeomorphisms. Thus, any
quantity of the form $\dot{T}_a[h]$ defines a gauge invariant
observables in linearized gravity, when Einstein or Einstein-matter
equations are linearly perturbed about a background solution $g$. A
straight forward argument shows that an IDEAL characterization provides
a list $\dot{T}_a[h]$, $a=1,\ldots,N$, of gauge invariant observables
that is also complete (it suffices to check that $T_a[g+h]$ do not
approach zero at $O(h^2)$ or higher order). That is, the joint kernel of
$\dot{T}_a[h]=0$ locally consists only of pure gauge modes
($h=\mathcal{L}_v g$ for some vector field $v$). The use of such local
observables (given by differential operators) can be advantageous both
in theoretical and practical investigations of classical and quantum
field theoretical models because they cleanly separate the local (or
ultraviolet) and global (or infrared) aspects of the theory. This is of
particular and current relevance to some controversies in inflationary
models of early universe cosmology~\cite{urakawa-tanaka,miao-woodard}.
Despite their importance, complete lists of (linearized) local gauge
invariant observables have been explicitly produced only in very few
cases, by ad-hoc methods. For instance, in the case of Einstein
equations coupled to a single inflaton field, a complete list has been
produced only recently~\cite{frob-hack-higuchi}. On the other hand,
linearising the equations of an IDEAL characterization provides a
systematic method of construction. The results of this method can be
compared to those of~\cite{frob-hack-higuchi} and are
equivalent~\cite{FHK}. Since these two sets of results naturally appear
in rather different forms, a detailed comparison is beyond the scope of
this work and will be presented elsewhere.

A similar geometric approach to the
construction of gauge invariant linearized observables was taken
in~\cite{EllisBruni}, using what we would call a partial IDEAL
characterization of cosmological spacetimes. No proof of their
completeness was ever given. In a sense, we complete the earlier
literature in this regard.

In this work, we add the cases of FLRW and inflationary spacetimes to
the (unfortunately still small) literature concerning IDEAL characterizations of
isometry classes of individual reference geometries. Other IDEAL
characterizations for geometries of interest in General Relativity
include Schwarzschild~\cite{Ferrando1998},
Reissner-Nordstr\"om~\cite{Ferrando2002}, Kerr~\cite{Ferrando2009},
Lema\^itre-Tolman-Bondi~\cite{Ferrando2010}, Stephani
universes~\cite{Ferrando2017} (see references for complete lists and
details) and of course the classic cases of constant curvature spaces,
which are known to be fully characterized by the structure of the
Riemann tensor (by theorems of Riemann and Killing-Hopf).

The synopsis of the paper is the following: In Section \ref{sub_main_results} 
we fix our notation and we outline our main 
results on the IDEAL characterization of FLRW spacetimes
(Theorem~\ref{thm_FLRW_class}) and inflationary spacetimes
(Theorem~\ref{thm_infl_class}). Our main goal there is to discuss our findings without dwelling on the technical proofs, which are left to the next sections. Hence a reader who wishes to focus more on the physical aspects of this paper should refer mainly to this part of the paper. In addition, still in Section \ref{sub_main_results}, we provide flowcharts for classifying spacetimes into FLRW and inflationary isometry classes, visually summarizing the contents of
Theorems~\ref{thm_FLRW_class} and~\ref{thm_infl_class}. In
Section~\ref{sec_FLRW_geom} we collect relevant information on the
geometry of FLRW and inflationary spacetimes. In
Section~\ref{sec_ideal}, we distinguish the possible local isometry
classes of FLRW or inflationary geometries and prove our main theorems.

\subsection{Main Results}\label{sub_main_results}

In this subsection, our goal is to introduce our conventions and to outline our main results. Therefore we will not dwell on the mathematical proofs, but we focus on the basic technical tools, necessary to formulate and to understand the physical significance of our findings. 

In this work, a \emph{spacetime} or \emph{Lorentzian manifold} $(M,g)$ will be a
smooth finite dimensional manifold $M$ (also Hausdorff, second
countable, connected and orientable) of $\dim M = m+1 \ge 2$, with a
Lorentzian metric $g$ (with signature ${-}{+}\cdots{+}$). A
\emph{spacetime with scalar} will consist of a triple $(M,g,\phi)$,
where $(M,g)$ is a Lorentzian manifold and $\phi\colon M \to \mathbb{R}$
is a smooth scalar field. Obviously, we could always consider the
spacetime $(M,g)$ as the special spacetime with zero scalar, $(M,g,0)$. 
In addition, with inflationary spacetimes, we will be assuming that the metric and the scalar field satisfy the coupled Einstein-Klein-Gordon equations, possibly with a nonlinear potential.

\noindent These observations should be kept in mind while reading the following
\begin{defn}[locally isometric] \label{def_loc_isom}
A spacetime with scalar $(M_1,g_1,\phi_1)$ is \emph{locally isometric at
$x_1\in M_1$ to} a spacetime with scalar $(M_2,g_2,\phi_2)$ \emph{at
$x_2 \in M$} if there exist open neighbourhoods $U_1 \ni x_1$, $U_2 \ni
x_2$ and a diffeomorphism $\chi\colon U_1 \to U_2$ such that $\chi(x_1)
= x_2$, $\chi^* g_2 = g_1$ and $\chi^*\phi_2 = \phi_1$. If we can choose
$U_1 = M_1$ and $U_2 = M_2$ then they are \emph{(globally) isometric}.
If for every $x_1\in M$ there is $x_2 \in M_2$ such that
$(M_1,g_1,\phi_1)$ at $x_1$ is locally isometric to $(M_2,g_2,\phi_2)$
at $x_2$, we simply say that $(M_1,g_1,\phi_1)$ is \emph{locally
isometric to} $(M_2,g_2,\phi_2)$ (note the asymmetry in the definition).
If $(M_1,g_1,\phi_1)$ is locally isometric to $(M_2,g_2,\phi_2)$, as
well as vice versa, we say that they are \emph{locally isometric to each
other} (which constitutes an equivalence relation). All spacetimes with
scalar that are locally isometric to a reference $(M,g,\phi)$ constitute
its \emph{local isometry class}.
\end{defn}

Our main results give an IDEAL characterization of local isometry
classes of \emph{regular} FLRW and inflationary spacetimes. In the following we give their precise definition, which is motivated in more detail in Sections~\ref{sec_perfect_fluid} and~\ref{sec_reg_FLRW}. Starting from the first case:

\begin{defn}[regular FLRW spacetime] \label{def_reg_FLRW}
	Let us fix a constant $\kappa \ne 0$. Denote by the triple
	$(m,\alpha,f)$, of a dimension $m\ge 1$, a constant $\alpha \in
	\mathbb{R}$ and a smooth positive function $f\colon I\to \mathbb{R}$
	defined on an interval $I \subseteq \mathbb{R}$, the corresponding FLRW
	spacetime $(M,g) = (I\times F, -dt^2 + f^2 g^F)$
	(Definition~\ref{def_FLRW}), with $\alpha$ the sectional curvature of
	$(F,g^F)$ and $F \cong S^m$ (when $\alpha>0$) or $F \cong \mathbb{R}^m$
	(when $\alpha \le 0$).
	
	We call $(M,g)$ a \emph{regular FLRW spacetime} if it belongs to one of
	the pa\-ra\-met\-rized families identified below.
	
	\begin{enumerate}
		\item
		Constant curvature spacetime, with spacetime sectional curvature $K$:
		\begin{equation}
		\CC^m_{K} =
		\begin{cases}
		\{ (m,K,\cosh(\sqrt{K}t)) & \mid K>0, ~~ I=\mathbb{R} \} , \\
		\{ (m,0,1) & \mid K=0, ~~ I=\mathbb{R} \} , \\
		\{ (m,K,\cos(\sqrt{-K}t)) & \mid K<0, ~~ I=\mathbb{R} \} .
		\end{cases}
		\end{equation}
		
		\item
		Einstein static universe, with spatial sectional curvature $K\ne 0$:
		\begin{equation}
		\ESU^m_{K} = \{ (m,K,1) \mid m>1, ~~ I=\mathbb{R} \} .
		\end{equation}
		
		\item
		Spatially flat constant scalar curvature spacetime, with spacetime
		scalar curvature $m(m+1)K$ and such that $\frac{f'^2}{f^2}(I) = J$:
		\begin{multline}
		\CSC^{m,0}_{K,J} = \left\{ (m,0,f) \mid m>1, ~ f'\ne 0,
		\phantom{\frac{f'}{f}}\right.
		\\ \left.
		\left(\frac{f''}{f} - \frac{f'^2}{f^2}\right)
		+ \frac{(m+1)}{2} \left(\frac{f'^2}{f^2} - K\right) = 0
		\right\} .
		\end{multline}
		
		\item
		Generic constant scalar curvature spacetime, with spacetime scalar
		curvature $m(m+1)K$, normalized radiation density constant $\Omega$ and
		such that $\frac{\alpha}{f^2}(I) = J$:
		\begin{multline}
		\CSC^m_{K,\Omega,J} = \left\{ (m,\alpha,f) \mid
		m>1, ~ \alpha \ne 0, ~ f' \ne 0,
		\phantom{\frac{f'}{f}} \right.
		\\ \left.
		\frac{f'^2}{f^2} + \frac{\alpha}{f^2} =
		K + \Omega \frac{|\alpha|^{(m+1)/2}}{f^{m+1}}
		\right\} .
		\end{multline}
		
		\item
		Spatially flat FLRW spacetime with normalized pressure function $P$
		defined on an open interval $J$, with $0 < \frac{f'^2}{f^2}(I) = J$ and
		\begin{equation}
		P(u) \left[\del_u P(u) - \frac{1}{2\kappa}\right] \ne 0
		\end{equation}
		everywhere on $J$:
		\begin{equation}
		\FLRW^{m,0}_{P,J} = \left\{ (m,0,f)
		\mid
		\left(\frac{f''}{f}-\frac{f'^2}{f^2}\right) + \frac{m}{2}\frac{f'^2}{f^2}
		= -\kappa P\left((f'/f)^2\right) \right\} .
		\end{equation}
		
		\item
		Generic FLRW spacetime with normalized energy function $E$ defined on an open
		interval $J$, with $0 \not\in \frac{\alpha}{f^2}(I) = J$ and
		\begin{equation}
		\del_u \left[ u \del_u E(u) - \frac{(m+1)}{2} \right] \ne 0
		\end{equation}
		everywhere on $J$:
		\begin{equation}
		\FLRW^m_{E,J} = \left\{ (m,\alpha,f) \mid
		m>1, ~ \alpha \ne 0, ~
		\frac{f'^2}{f^2} + \frac{\alpha}{f^2} = \kappa E(\alpha/f^2)
		\right\} .
		\end{equation}
	\end{enumerate}
\end{defn}

\noindent Next, we focus our attention to the inflationary spacetimes, following the more detailed motivation from Sections~\ref{sec_scalar} and~\ref{sec_reg_infl}:

\begin{defn}[regular inflationary spacetime] \label{def_reg_infl}
	Let us fix a constant $\kappa \ne 0$. Denote by the quadruple
	$(m,\alpha,f,\phi)$, of dimension $m>1$, constant $\alpha\in
	\mathbb{R}$, and smooth functions $f,\phi\colon I \to \mathbb{R}$
	defined on an interval $I \subseteq \mathbb{R}$, with $f$ positive, the
	corresponding inflationary spacetime $(M,g,\phi) = (I\times F, -dt^2 +
	f^2 g^F, \bar{\phi})$ (Definition~\ref{def_infl}), with $\bar{\phi}$
	being the composition of standard projection $I\times F \to I$ with
	$\phi$, with $\alpha$ the sectional curvature of $(F,g^F)$ and $F\cong
	S^m$ (when $\alpha > 0$) or $F \cong \mathbb{R}^m$ (when $\alpha \le
	0$).
	
	We call $(M,g,\bar{\phi})$ a \emph{regular inflationary spacetime} if it
	belongs to one of the parametrized families identified below.
	\begin{enumerate}
		\item
		Constant scalar, with scalar value $\Phi$, on a constant curvature
		spacetime with scalar curvature $K$:
		\begin{equation}
		\CC^m_K \CS_\Phi = \{ (m,\alpha,f,\Phi)
		\mid (m,\alpha,f) \in \CC^m_K \} .
		\end{equation}
		\item
		Constant energy scalar, with energy density $\rho > 0$ and $J =
		\phi(I)$, on an Einstein static universe with spatial sectional
		curvature $K = \frac{2}{m(m-1)} \kappa \rho$, or equivalently with
		cosmological constant $\Lambda = \frac{(m-1)}{m} \kappa \rho$:
		\begin{equation}
		\ESU^m_{K} \CES_{\rho,J}
		= \{ (m,K,1, \sqrt{2\rho/m} t) \mid I = J/\sqrt{2\rho/m} \} .
		\end{equation}
		\item
		Spatially flat massless minimally-coupled scalar spacetime, with
		cosmological constant $\Lambda$, $J = \phi(I)$ and $J' = \frac{f'}{f}(I)
		\not\ni 0$ and $\frac{2\Lambda/\kappa}{m(m-1)} < \frac{1}{\kappa}
		(J')^2$:
		\begin{multline} \textstyle
		\MMS^{m,0}_{\Lambda,J,J'} = \left\{ (m,0,f,\phi) \mid
		\phi' < 0 , ~~
		\frac{f'}{f} \ne 0 , 
		\right. \\ \textstyle
		\left.
		\frac{f'^2}{f^2} = \frac{\kappa \phi'^2 + 2\Lambda}{m(m-1)} , ~~
		\left(\frac{f''}{f} - \frac{f'^2}{f^2}\right) + m\frac{f'^2}{f^2} =
		\frac{2\Lambda}{(m-1)} \right\} .
		\end{multline}
		\item
		Generic massless minimally-coupled scalar spacetime, with
		cosmological constant $\Lambda$, normalized scalar energy constant
		$\Omega > 0$, $J = \phi(I)$ and $J' = \frac{f'}{f}(I) \not\ni 0$:
		\begin{multline} \textstyle
		\MMS^m_{\Lambda,\Omega,J,J'}
		= \left\{ (m,\alpha,f,\phi) \mid \alpha \ne 0, ~~
		\frac{f'}{f} \ne 0 , 
		\right. \\ \textstyle
		\left.
		\phi' = -\sqrt{\Omega} \frac{|\alpha|^{\frac{m}{2}}}{f^m}, ~~
		\frac{f'^2}{f^2} + \frac{\alpha}{f^2}
		= \frac{2\Lambda + \kappa \Omega |\alpha|^m/f^{2m}}{m(m-1)}
		\right\} .
		\end{multline}
		\item
		Spatially flat nonlinear Klein-Gordon spacetime, with non-constant
		scalar self-coupling potential $V\colon J\to \mathbb{R}$, with $J =
		\phi(I)$, and expansion profile $\Xi\colon J \to \mathbb{R}$,
		satisfying $\Xi(u) \ne 0$, $\frac{1}{\kappa}\del_u \Xi(u) > 0$ and
		$\mathfrak{H}_V(\Xi) = 0$ in the notation of~\eqref{eq_xi_flat}:
		\begin{equation} \textstyle
		\NKG^{m,0}_{V,\Xi,J} = \left\{ (m,0,f,\phi) \mid
		\frac{f'}{f} = \Xi(\phi) ,
		\phi' = -\frac{(m-1)}{\kappa} \del_\phi \Xi(\phi)
		\right\} .
		\end{equation}
		\item
		Generic nonlinear Klein-Gordon spacetime, with non-constant scalar
		potential $V\colon J \to \mathbb{R}$, with $J = \phi(I)$, and
		expansion profile $(\Pi,\Xi) \colon J\to \mathbb{R}^2$, satisfying
		$\Pi < 0$, $\Xi \ne 0$, $\kappa\frac{\Pi^2+V}{m(m-1)} \ne \Xi^2$ and
		$\mathfrak{G}_V(\Pi,\Xi) = 0$ in the notation of~\eqref{eq_pi_xi}:
		\begin{multline} \textstyle
		\NKG^m_{V,\Pi,\Xi,J} = \left\{ (m,\alpha,f,\phi) \mid
		\alpha \ne 0 , ~~ \frac{f'}{f} \ne 0 , \right. \\ \textstyle
		\left.
		\phi' = \Pi(\phi) , ~~
		\frac{f'}{f} = \Xi(\phi) , ~~
		\frac{f'^2}{f^2} + \frac{\alpha}{f^2}
		= \kappa \frac{\phi'^2 + V(\phi)}{m(m-1)}
		\right\} .
		\end{multline}
	\end{enumerate}
\end{defn}

Below we directly give the list of tensor equations, covariantly constructed from the metric, the Riemann curvature, and its derivatives, that characterize
the corresponding local isometry classes. Observe that an IDEAL
characterization is not unique. Given one, many others can be produced
by covariant and invertible transformations. Our choices are based on
various conventions used in relativity and cosmology.

To be specific, our conventions for the relations between the metric
$g_{ij}$, covariant derivative $\nabla_i$, Riemann curvature, as well as
Ricci tensor and scalar are the following:
\begin{gather*}
	(\nabla_i\nabla_j - \nabla_j\nabla_i) v_k = R_{ijk}{}^l v_l , \quad
	R_{ijkh} = R_{ijk}{}^l g_{lh} , \\
	R_{ik} = R_{ikj}{}^k , \quad
	\mathcal{R} = R_{ij} g^{ij} , \quad
	\mathcal{B} = R_{ij} R_{kl} g^{ik} g^{jl} .
\end{gather*}
It is convenient to define the following product (sometimes also known
as the \emph{Kulkarni-Nomizu product}) that builds an object with the
symmetries of the Riemann tensor out of symmetric $2$-tensors $A_{ij}$,
$B_{ij}$:
\begin{equation}\label{kn-product}
	(A \odot B)_{ijkh}
	= A_{ik} B_{jh} - A_{jk} B_{ih} - A_{ih} B_{jk} + A_{jh} B_{ik} .
\end{equation}
Note that $A\odot B = B\odot A$ and $UU \odot UU = 0$, with $(UU)_{ij} =
U_i U_j$ and $U_i$ any vector field. For $\dim M = m+1>2$, our formula
for the Weyl tensor is
\begin{equation}
	W_{ijkh} = R_{ijkh} - \frac{1}{(m-1)} (g\odot R)_{ijkh}
		+ \frac{1}{2m(m-1)} \mathcal{R} (g\odot g)_{ijkh} .
\end{equation}
Note that $W_{ijkh}$ vanishes precisely when $R_{ijkh} = (g\odot
A)_{ijkh}$ for some symmetric $A_{ij}$. As usual, we denote idempotent
symmetrization and anti-symmetrization by $A_{(ij)} = \frac{1}{2}
(A_{ij} + A_{ji})$, $A_{[ij]} = \frac{1}{2} (A_{ij} - A_{ji})$.

The first theorem classifies just the Lorentzian spacetime, without
reference to a scalar field.
The definitions for the various scalar and tensor fields introduced below
may seem ad-hoc, but they have straightforward geometrical meanings.
The vector field $U^i$ plays the role of a future-pointing unit timelike
vector field, orthogonal to the cosmological spatial slices. It is
defined as a normalized gradient of a curvature scalar, with the choice
of curvature scalar depending on the precise case being considered. The
tensors $\mathfrak{P}_{ij}$ and $\mathfrak{D}_{ij}$ encode in them the
shear, twist and geodesic character of $U^i$ and are non-zero when the
spacetime deviates from a generalized Robertson-Walker (GRW) spacetime
(a possibly non-homogeneous geometry undergoing cosmological expansion
or contraction). The expansion $\xi$ of the vector field $U^i$ also
plays the role of the Hubble rate, while $\mathbf{\eta}$ that of the
Hubble acceleration. The tensor $\mathfrak{E}_{ijkh}$ measures the
deviation from the spatial slices from homogeneity and isotropy, while
the scalar $\zeta$, together with $\mathfrak{E}_{ijkh}$, measures the
deviation of spatial slices from flatness.
\begin{thm} \label{thm_FLRW_class}
Consider a Lorentzian manifold $(M,g)$ of $\dim M = m+1 \ge 2$,
$\kappa\ne 0$ a fixed constant. With $U$ a unit timelike vector field,
consider the following notations, which are defined when possible:
\begin{equation}
\begin{gathered}
	\xi := \frac{\nabla^i U_i}{m} , \quad
	\bm{\eta} := U^i \nabla_i \xi , \quad 
	\zeta := \frac{\mathcal{R} - 2m\bm{\eta}
		- \frac{1}{2}m(m+1) \xi^2}{m(m-1)} , \\
	\mathfrak{P}_{ij} := U_{[i} \nabla_{j]} \xi , \quad
	\mathfrak{D}_{ij} := \nabla_i U_j
		- \frac{\nabla^k U_k}{m} (g_{ij} + U_i U_j) , \\
	\mathfrak{Z}_{ijkh} := R_{ijkh}
	- \left( g \odot \left[ \frac{\xi^2}{2} g
		- \bm{\eta} UU \right] \right)_{ijkh} , \\
	\mathfrak{C}_{ijkh} := R_{ijkh}
	- \left( g \odot \left[ \frac{(\xi^2+\zeta)}{2} g
		- (\bm{\eta}-\zeta) UU \right] \right)_{ijkh} , \\
	U_\mathcal{R}
		:= \frac{-\nabla \mathcal{R}}{\sqrt{-(\nabla\mathcal{R})^2}} , \quad
	U_\mathcal{B}
		:= \frac{-\nabla \mathcal{B}}{\sqrt{-(\nabla\mathcal{B})^2}} .
\end{gathered}
\end{equation}

Given $x\in M$, Table~\ref{tab_FLRW_class} gives the list of
inequalities and equations (right column, written using the above
notation, with a specific choice of $U$) that are satisfied on a
neighborhood of $x$ if and only if the Lorentzian manifold belongs to
the corresponding local isometry class at $x$ (left column) of a regular
FLRW spacetime (Definition~\ref{def_reg_FLRW}). Each local isometry
class belongs to a family, parametrized by real constants, intervals or
functions (middle column). By continuity, each inequality need only be
checked at $x$.

\end{thm}

In addition, since 
both Theorem~\ref{thm_FLRW_class} and Table~\ref{tab_FLRW_class} are  densely packed with information, we include a graphical
flowchart summaries of the same information in
Figures~\ref{flowchart_FLRW_class}. The notation is the same as in the original theorems.

Finally, we state the theorem classifying inflationary spacetimes, those
endowed with scalar and satisfying the coupled Einstein-Klein-Gordon
equations, where the equation for the scalar $\phi$ may be nonlinear due
to a potential $V(\phi)$. The reader is referred to the paragraph
preceding Theorem~\ref{thm_FLRW_class} for an explanation of the
notation. The new scalar $\mathfrak{H}_V$ roughly corresponds to the
Hamilton-Jacobi equation of spatially flat single field inflation, while
$\mathfrak{G}_V$ is its generalization to the non-flat case. See the end
of Section~\ref{sec_reg_infl} for a more detailed motivation.
\begin{thm} \label{thm_infl_class}
Consider an inflationary spacetime $(M,g,\phi)$ of $\dim M = m+1 > 2$,
$\kappa \ne 0$ a fixed constant. With $U$ a unit timelike vector field,
recall the notation of Theorem~\ref{thm_FLRW_class}, supplemented with
\begin{gather}
	(-)' := U^i \nabla_i(-), \quad
	U_\phi := \frac{\nabla \phi}{\sqrt{-(\nabla \phi)^2}} , \\
\label{eq_xi_flat}
	\mathfrak{H}_V(\Xi) := \left(\del_u\Xi\right)^2
		- \kappa\frac{m \Xi^2}{(m-1)}
		+ \kappa^2\frac{V}{(m-1)^2} , \\
\label{eq_pi_xi}
	\mathfrak{G}_V(\Pi,\Xi) :=
	\begin{pmatrix}
		\Pi \left( \del_u \Xi + \kappa\frac{\Pi}{(m-1)}\right)
			- \left(\kappa\frac{\Pi^2 + V}{m(m-1)} - \Xi^2\right) \\
		\del_u \left(\kappa\frac{\Pi^2 + V}{m(m-1)} - \Xi^2\right)
			+ 2\frac{\Xi}{\Pi}
				\left(\kappa \frac{\Pi^2 + V}{m(m-1)} - \Xi^2\right)
	\end{pmatrix} ,
\end{gather}
where $\Xi = \Xi(u)$ and $\Pi = \Pi(u)$. Let $g$ and $\phi$ satisfy the
coupled Einstein-Klein-Gordon equations with scalar potential $V(\phi)$,
\begin{gather}
	\nabla^i \nabla_i \phi - \frac{1}{2} \del_\phi V(\phi) = 0 , \\
	R_{ij} - \frac{1}{2} \mathcal{R} g_{ij}
		= \kappa \left((\nabla_i\phi)(\nabla_j\phi)
			- \frac{1}{2} g_{ij} [(\nabla\phi)^2 + V(\phi)]\right) .
\end{gather}

Given $x\in M$, Table~\ref{tab_infl_class} gives the list of
inequalities and equations (right column) that are satisfied on some
neighborhood of $x$ if and only if the inflationary spacetime belongs to
the corresponding local isometry class at $x$ (left column) of a regular
inflationary spacetime. Each local isometry class belongs to a family,
parametrized by real constants, intervals or functions (middle column).
By continuity, each inequality needs only to be checked at $x$.
\end{thm}

In addition, since 
both Theorem~\ref{thm_infl_class} and Table~\ref{tab_infl_class} are  densely packed with information, we include a graphical
flowchart summaries of the same information in
Figures~\ref{flowchart_infl_class}. The notation is the same as in the original theorems.

Note that when we make the choice $U = U_\phi$, it automatically follows
that $\phi' = -\sqrt{-(\nabla\phi)^2} < 0$. This convention is common in
the study of inflation, where $\phi(t)$ starts off at a high value and
then ``rolls down hill'' as $t$ increases. This is reflected in the
inequalities in Table~\ref{tab_infl_class}.

Our characterization theorems cover what we have called \emph{regular}
FLRW or inflationary spacetimes (Definitions~\ref{def_reg_FLRW}
and~\ref{def_reg_infl}), which are required to satisfy the
inequalities listed in Tables~\ref{tab_FLRW_class}
and~\ref{tab_infl_class} everywhere.

\begin{table}
\begin{center}
	\begin{tabular}{lcc}
		class & parameters/$U$ & inequalities/equations
		\\ \hline\hline
		\multicolumn{2}{l}{\footnotesize\emph{(a) Constant Curvature}} \\
		$\CC^m_K$ & & $R_{ijkh} - \frac{K}{2} (g\odot g)_{ijkh} = 0$
		\\ \noalign{\smallskip}\hline
		\multicolumn{2}{l}{\footnotesize\emph{(b) Einstein Static Universe}} \\
		$\ESU^m_K$ &
		\multirow{2}{*}{$\begin{gathered} m>1, \\ K\ne 0 \end{gathered}$} &
		$\exists V^i\colon \left(g_{ij} - \frac{R_{ij}}{(m-1)K}\right) V^i V^j < 0$
		\\ \noalign{\smallskip}\cline{3-3}\noalign{\smallskip}
		& & 
		$\begin{gathered}
		W_{ijkh} = 0 , \quad
		\nabla_i R_{jk} = 0 , \\
		R_i^j \left(R_j^k - (m-1) K \delta_j^k\right) = 0 , \\
		\mathcal{R} - m(m-1) K = 0
		\end{gathered}$
		\\ \noalign{\smallskip}\hline
		\multicolumn{2}{l}{\footnotesize\emph{(c) Spatially Flat Constant Scalar Curvature}} \\
		$\CSC^{m,0}_{K,J}$ &
		\multirow{2}{*}{$\begin{gathered}
			m>1, \\
			0 < J \subset \mathbb{R} \\[1ex]
			(U = U_\mathcal{B})
			\end{gathered}$} &
		$(\nabla\mathcal{B})^2 < 0, \quad \xi^2(x) \in J$
		\\ \noalign{\smallskip}\cline{3-3}\noalign{\smallskip}
		& &
		$\begin{gathered}
		\mathfrak{P}_{ij} = 0 , \quad
		\mathfrak{D}_{ij} = 0 , \\
		\mathfrak{Z}_{ijkh} = 0 , \\ \textstyle
		\bm{\eta} + \frac{(m+1)}{2} (\xi^2 - K) = 0
		\end{gathered}$
		\\ \noalign{\smallskip}\hline
		\multicolumn{2}{l}{\footnotesize\emph{(d) Generic Constant Scalar Curvature}} \\
		$\CSC^m_{K,\Omega,J}$ &
		\multirow{2}{*}{$\begin{gathered}
			m>1, \quad \Omega \ne 0, \\
			0 \not\in J \subset \mathbb{R} \\[1ex]
			(U = U_\mathcal{B})
			\end{gathered}$} &
		$(\nabla\mathcal{B})^2 < 0, \quad \zeta(x) \in J$
		\\ \noalign{\smallskip}\cline{3-3}\noalign{\smallskip}
		& &
		$\begin{gathered}\textstyle
		\nabla_i U_j -  \frac{\nabla_i \zeta}{2\zeta} U_j - \xi g_{ij} = 0 , \\
		\mathfrak{C}_{ijkh} = 0 , \\
		\xi^2 + \zeta - K - \Omega |\zeta|^{\frac{m+1}{2}} = 0
		\end{gathered}$
		\\ \noalign{\smallskip}\hline
		\multicolumn{2}{l}{\footnotesize\emph{(e) Spatially Flat FLRW}} \\
		$\FLRW^{m,0}_{P,J}$ &
		\multirow{2}{*}{$\begin{gathered}
			0 < J \subset \mathbb{R}, \\
			P\colon J\to \mathbb{R}, \\ \textstyle
			P[\del_u P - \frac{1}{2\kappa}] \ne 0 \\[1ex]
			(U = U_\mathcal{R})
			\end{gathered}$} &
		$(\nabla\mathcal{R})^2 < 0, \quad
		\bm{\eta} \ne 0, \quad \xi^2(x) \in J$
		\\ \noalign{\smallskip}\cline{3-3}\noalign{\smallskip}
		& &
		$\begin{gathered}
		\mathfrak{P}_{ij} = 0 , \quad
		\mathfrak{D}_{ij} = 0 , \\
		\mathfrak{Z}_{ijkh} = 0 , \\
		\bm{\eta} + \frac{m}{2}\xi^2 + \kappa P(\xi^2) = 0
		\end{gathered}$
		\\ \noalign{\smallskip}\hline
		\multicolumn{2}{l}{\footnotesize\emph{(f) Generic FLRW}} \\
		$\FLRW^m_{E,J}$ &
		\multirow{2}{*}{$\begin{gathered}
			m>1, \quad 0 \not\in J \subset \mathbb{R} , \\
			E\colon J \to \mathbb{R}, ~ \kappa E(u) > u , \\ \textstyle
			\del_u[u\del_u E - \frac{(m+1)}{2}E] \ne 0 \\[1ex]
			(U = U_\mathcal{R})
			\end{gathered}$} &
		$(\nabla\mathcal{R})^2 < 0, \quad
		\xi \ne 0, \quad \zeta(x) \in J$
		\\ \noalign{\smallskip}\cline{3-3}\noalign{\smallskip}
		& &
		$\begin{gathered}\textstyle
		\nabla_i U_j -  \frac{\nabla_i \zeta}{2\zeta} U_j - \xi g_{ij} = 0 , \\
		\mathfrak{C}_{ijkh} = 0 , \\
		\xi^2+\zeta - \kappa E(\zeta) = 0
		\end{gathered}$
		\\ \noalign{\smallskip}\hline\hline
	\end{tabular}
\end{center}
\caption{IDEAL characterization of local isometry classes of regular
	FLRW spacetimes (Theorem~\ref{thm_FLRW_class}).\label{tab_FLRW_class}}
\end{table}

\begin{table}
\begin{center}
\begin{tabular}{lcc}
	class & parameters & inequalities/equalities
	\\ \hline\hline
	\multicolumn{3}{l}{\footnotesize\emph{(a) Constant Scalar}} \\
	$\CC^m_K \CS_\Phi$ &
		$\begin{gathered} \textstyle
			V(u) = \frac{2}{\kappa} \Lambda , \\ \textstyle
			K = \frac{2}{m(m-1)} \Lambda
		\end{gathered}$ &
		$\begin{gathered} \textstyle
			R_{ijkh} - \frac{\Lambda}{m(m-1)} (g\odot g)_{ijkh} = 0 , \\
			\phi = \Phi
		\end{gathered}$
	\\ \noalign{\smallskip}\hline
	\multicolumn{3}{l}{\footnotesize\emph{(b) Constant Energy Scalar}} \\
	$\ESU^m_K \CES_{\rho,J}$ &
		\multirow{2}{*}{$\begin{gathered} \textstyle
			V(u)=\frac{2(m-1)}{m}\rho , \\ \textstyle
			\rho>0, ~ K = \frac{2\kappa}{m(m-1)} \rho \\[1ex]
			(U = U_\phi)
		\end{gathered}$} &
		$(\nabla\phi)^2 < 0, \quad
			\frac{\zeta}{\kappa} > 0, \quad
			\phi(x) \in J$
	\\ \noalign{\smallskip}\cline{3-3}\noalign{\smallskip}
	& &
		$\begin{gathered}\textstyle
			\nabla_i U_j = 0 , ~~ \mathfrak{C}_{ijkh} = 0 , \\ \textstyle
			(\nabla\phi)^2 = -\frac{2}{m} \rho , ~~
			\zeta = \frac{2\kappa\rho}{m(m-1)}
		\end{gathered}$
	\\ \noalign{\medskip}\hline
	\multicolumn{3}{l}{\footnotesize\emph{(c) Spatially Flat Massless Minimally-coupled Scalar}} \\
	$\MMS^{m,0}_{\Lambda,J,J'}$ &
		\multirow{2}{*}{$\begin{gathered}
			0 \not\in J' , \\  \textstyle
			\frac{2\Lambda/\kappa}{m(m-1)} < \frac{1}{\kappa} (J')^2 \\ \textstyle
			V(u) = \frac{2}{\kappa} \Lambda \\[1ex]
			(U = U_\phi)
		\end{gathered}$} &
		$\begin{gathered} \textstyle
			(\nabla\phi)^2 < 0 ,
			\frac{1}{\kappa} (\xi^2-\frac{2\Lambda}{m(m-1)}) > 0 , \\ \textstyle
			\phi(x) \in J , ~~
			\xi(x) \in J'
		\end{gathered}$
	\\ \noalign{\smallskip}\cline{3-3}\noalign{\smallskip}
	& &
		$\begin{gathered}\textstyle
			\nabla_i U_j - \frac{\nabla_i \phi'}{m\phi'} U_j - \xi g_{ij} = 0 , \\
			\textstyle
			\mathfrak{Z}_{ijkh} = 0 , ~~
			\bm{\eta} + m \xi^2 = \frac{2\Lambda}{(m-1)} , \\ \textstyle
			\xi^2 = \frac{\kappa \phi'^2 + 2\Lambda}{m(m-1)}
		\end{gathered}$
	\\ \noalign{\smallskip}\hline
	\multicolumn{3}{l}{\footnotesize\emph{(d) Generic Massless Minimally-coupled Scalar}} \\
	$\MMS^m_{\Lambda,\Omega,J,J'}$ &
		\multirow{2}{*}{$\begin{gathered} \textstyle
			V(u) = \frac{2}{\kappa} \Lambda , \\
			0\not\in J', ~~ \Omega > 0 \\[1ex]
			(U = U_\phi)
		\end{gathered}$} &
		$\begin{gathered}
			(\nabla\phi)^2 < 0 , \\
			\phi(x) \in J, ~~
			\xi(x) \in J'
		\end{gathered}$
	\\ \noalign{\smallskip}\cline{3-3}\noalign{\smallskip}
	& &
		$\begin{gathered}\textstyle
			\nabla_i U_j - \frac{\nabla_i \phi'}{m\phi'} U_j - \xi g_{ij} = 0 , \\
			\mathfrak{C}_{ijkh} = 0 , ~~
			\phi' = -\sqrt{\Omega} |\zeta|^{\frac{m}{2}} , \\ \textstyle
			\xi^2 + \zeta = \frac{2\Lambda + \kappa \Omega |\zeta|^m}{m(m-1)}
		\end{gathered}$
	\\ \noalign{\smallskip}\hline
	\multicolumn{2}{l}{\footnotesize\emph{(e) Spatially Flat Nonlinear Klein-Gordon}} \\
	$\NKG^{m,0}_{V,\Xi,J}$ &
		\multirow{3}{*}{$\begin{gathered}
			V,\Xi\colon J \to \mathbb{R} , \\
			\Xi(u) \ne 0,
			\frac{1}{\kappa} \Xi'(u) > 0 , \\
			V'(u) \ne 0 ,
			\mathfrak{H}_V(\Xi) = 0 \\[1ex]
			(U = U_\phi)
		\end{gathered}$} &
		$(\nabla\phi)^2 < 0, ~~ \xi \ne 0, ~~ \frac{1}{\kappa} \bm{\eta} < 0$
	\\ \noalign{\smallskip}\cline{3-3}\noalign{\smallskip}
	& &
		$\begin{gathered}
			\mathfrak{P}_{ij} = 0 , ~~
			\mathfrak{D}_{ij} = 0 , ~~
			\mathfrak{Z}_{ijkh} = 0 , \\ \textstyle
			\phi' = -\frac{(m-1)}{\kappa} \del_\phi \Xi(\phi) , \\
			\xi = \Xi(\phi)
		\end{gathered}$
	\\ \noalign{\bigskip}\hline
	\multicolumn{3}{l}{\footnotesize\emph{(f) Generic Nonlinear Klein-Gordon}} \\
	$\NKG^m_{V,\Pi,\Xi,J}$ &
		\multirow{2}{*}{$\begin{gathered}
			V,\Xi,\Pi\colon J \to \mathbb{R} , \\
			\Pi < 0 , ~~ \Xi \ne 0, \\ \textstyle
			\kappa \frac{\Pi^2+V}{m(m-1)} \ne \Xi^2 , \\
			V'(u) \ne 0 ,
			\mathfrak{G}_V(\Pi,\Xi) = 0 \\[1ex]
			(U = U_\phi)
		\end{gathered}$} &
		$\begin{gathered}
			(\nabla\phi)^2 < 0, ~~
			\xi \ne 0, \\ \textstyle
			\zeta \ne 0, ~~
			\frac{\bm{\eta} - \zeta}{\kappa} < 0
		\end{gathered}$
	\\ \noalign{\smallskip}\cline{3-3}\noalign{\smallskip}
	& &
		$\begin{gathered}
			\mathfrak{P}_{ij} = 0 , ~~
			\mathfrak{D}_{ij} = 0 , ~~
			\mathfrak{C}_{ijkh} = 0 , \\ 
			\zeta = \kappa \frac{\phi'^2 + V(\phi)}{m(m-1)} - \xi^2 , \\
			\phi' = \Pi(\phi) , ~~
			\xi = \Xi(\phi)
		\end{gathered}$
	\\ \noalign{\smallskip}\hline\hline
\end{tabular}
\end{center}
\caption{IDEAL characterization of local isometry classes of regular
inflationary spacetimes (Theorem~\ref{thm_infl_class}).\label{tab_infl_class}}
\end{table}

\begin{figure}[h!]
	\begin{footnotesize}
		\begin{center}
			\begin{tikzpicture}[auto]
			\node[rectangle, draw,
			text centered, minimum height=24pt](init){$(M,g), \; \dim M = m+1$};
			\node[signal, signal to=east and west, draw, text centered, node distance=1.2cm, below of=init](Ricci){$(\nabla\mathcal{R})^2 < 0$};
			\node[signal, signal to=east and west, draw, text centered, node distance= 1.2cm, below of=Ricci](Ricci2){$(\nabla\mathcal{B})^2 < 0$};
			\node[signal, signal to=east and west, draw, text centered, node distance= 3.3cm,  right of=Ricci2](CCCdesitter){$R - \frac{K}{2} (g\odot g)= 0$};
			\node[node distance=1.6cm, right of=CCCdesitter](pointright){};
			\node[node distance=1.6cm, left of=CCCdesitter](pointleft){};
			\node[node distance=1.7cm, left of=Ricci2](pointleftright){};
			\node[node distance=4.9cm, left of=Ricci2](pointleftleft){};
			\node[rectangle, draw, text width=6em, text centered, minimum height=24pt, below of=pointright, node distance= 1.2cm](Desitter){ $CC_K^m$};
			\node[rectangle, draw, text width=7em, text centered, minimum height=24pt, below of=Ricci2, node distance= 1.5cm](UB){$U:= U_{\mathcal{B}}$};
			\node[rectangle, draw, text width=7em, text centered, minimum height=24pt, left of=UB, node distance= 3.5cm](UR){ $U:= U_{\mathcal{R}}$};
			\node[signal, signal to=east and west, draw, text centered, node distance=3cm, text width=13em, below of=CCCdesitter](ifESU){$W_{ijkh} = 0 $\\
				$R_i^j \left(R_{jk} - (m-1) K g_{jk}\right) = 0 $ \\
				$\mathcal{R} - m (m-1) K = 0$ \\
				$\nabla_i R_{jk} = 0$ };
			\node[signal, signal to=east and west, draw, text centered, node distance=4.3cm, below of=UB](main){$\zeta=0$};
			\node[rectangle, draw, text width=4em, text centered, minimum height=24pt, below of=main, node distance= 3cm](nothing){not FLRW};
			\node[rectangle, draw, text width=6em, text centered, minimum height=24pt, below of=pointright, node distance= 4.8cm](ESU){$ESU_K^m$};
			\node[rectangle, draw, text width=4em, text centered, minimum height=24pt, below of=pointleft, node distance= 4.8cm](nothing2){not FLRW};
			\node[signal, signal to=east and west, draw, text width=10em, text centered, node distance=6cm, below of=UR](constant){$\mathfrak{P}_{ij} = 0, \quad \mathfrak{D}_{ij} = 0$\\ $\mathfrak{Z}_{ijkh} =0$ };	
			\node[signal, signal to=east and west, draw, text width=10em, text centered, node distance=4.5cm, below of=ifESU](PD){$\nabla_i U_j - \frac{\nabla_i \zeta}{2\zeta}U_j = \xi g_{ij}$\\$\mathfrak{C}_{ijkh}=0$};
			\node[signal, signal to=east and west, draw, text width=10em, text centered, node distance=2.7cm,  below of=constant](FLRW){$\bm{\eta}+\frac{m}{2}\xi^2=-\kappa P(\xi^2)$\\ $\xi\in J$};
			\node[signal, signal to=east and west, draw, text width=11em, text centered, node distance=1.3cm,  below of=FLRW](FLRW-P){$\kappa P(u) = \frac{1}{2} [u-(m+1)K]$};
			\node[signal, signal to=east and west, draw, text width=10em, text centered, node distance=2.7cm, below of=PD](PD2){$\xi^2+ \zeta= \kappa E(\zeta)$\\ $\zeta\in J$};
			\node[signal, signal to=east and west, draw, text width=10em, text centered, node distance=1.3cm, below of=PD2](PD2-E){$\kappa E(u) = K + \Omega |u|^{\frac{m+1}{2}}$};
			\node[rectangle, draw, text width=6em, text centered, minimum height=24pt, below of=pointleftleft, node distance= 12.7cm](flatFLRW){$\FLRW^{m,0}_{P,J}$ };
			\node[rectangle, draw, text width=6em, text centered, minimum height=24pt, below of=pointleftright, node distance= 12.7cm](flatFLRW-J){$\CSC^{m,0}_{K,J}$ };
			\node[rectangle, draw, text width=6em, text centered, minimum height=24pt, below of=pointleft, node distance= 12.7cm](genFLRW){$\FLRW^m_{E,J}$ };
			\node[rectangle, draw, text width=6em, text centered, minimum height=24pt, below of=pointright, node distance= 12.7cm](genFLRW-J){$\CSC^m_{K,\Omega,J}$ };
			\path[draw, -latex'] (init)--(Ricci);
			\path[draw, -latex'] (Ricci)--node {no}(Ricci2);
			\path[draw, -latex'] (Ricci2)--node {no}(CCCdesitter);
			\path[draw, -latex'] (CCCdesitter)--node [pos=.8]{yes}(Desitter);
			\path[draw, -latex'] (CCCdesitter)--node {no}(ifESU);
			\path[draw, -latex'] (Ricci)-|node[pos=.2] {yes}(UR);
			\path[draw, -latex'] (Ricci2)--node {yes}(UB);
			\path[draw, -latex'] (UR)|-(main);
			\path[draw, -latex'] (UB)--(main);
			\path[draw, -latex'] (PD)--node {no}(nothing);
			\path[draw, -latex'] (constant)--node [swap]{no}(nothing);
			\path[draw, -latex'] (ifESU)--node[left] {no} (nothing2);
			\path[draw, -latex'] (ifESU)--node[pos=1] {yes}(ESU);
			\path[draw, -latex'] (main)--node [swap]{yes}(constant);
			\path[draw, -latex'] (main)--node {no}(PD);
			\path[draw, -latex'] (PD)--node {yes}(PD2);
			\path[draw, -latex'] (constant)--node {yes}(FLRW);
			\path[draw, -latex'] (PD2)--node [swap]{no}(nothing);
			\path[draw, -latex'] (PD2)--node {yes}(PD2-E);
			\path[draw, -latex'] (PD2-E)--node[pos=.8] {yes}(genFLRW-J);
			\path[draw, -latex'] (PD2-E)--node [left]{no}(genFLRW);
			\path[draw, -latex'] (FLRW)--node {no}(nothing);
			\path[draw, -latex'] (FLRW)--node {yes}(FLRW-P);
			\path[draw, -latex'] (FLRW-P)--node[pos=.8] {yes}(flatFLRW-J);
			\path[draw, -latex'] (FLRW-P)--node [left]{no}(flatFLRW);
			\end{tikzpicture}		
			\caption{IDEAL characterization of local isometry classes of regular FLRW spacetimes (Theorem~\ref{thm_FLRW_class}, Table~\ref{tab_FLRW_class}).\label{flowchart_FLRW_class}}  
		\end{center}   
	\end{footnotesize}   
\end{figure} 

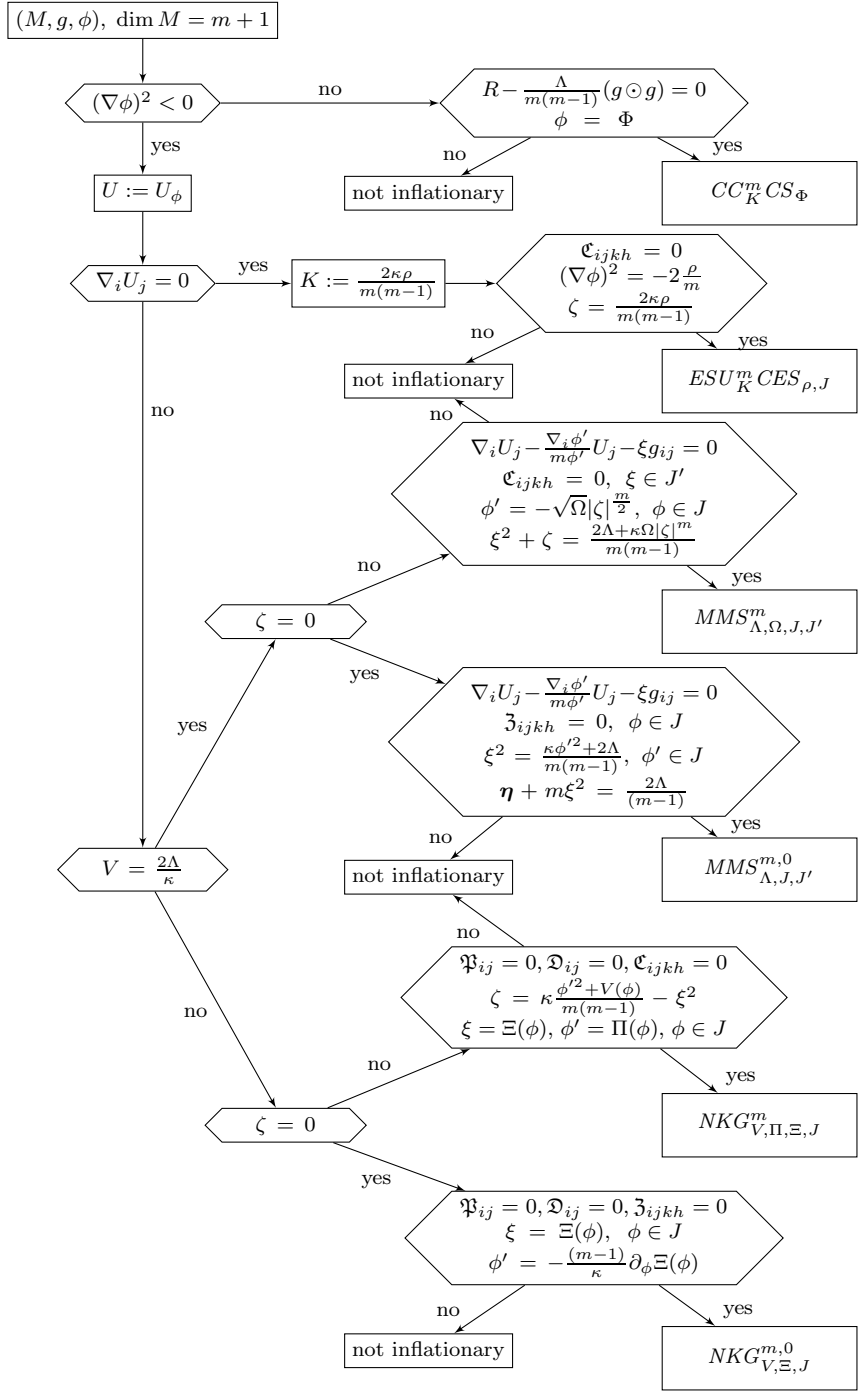
\begin{figure}[p]
	\begin{footnotesize}
		\begin{center}
			\begin{tikzpicture}[auto]
			\node[rectangle, draw, text centered, minimum height=12pt](first){$(M,g, \phi), \; \dim M = m+1$};
			\node[signal, signal to=east and west, draw, text centered, node distance=1.1cm, below of=first](nablaphi){$(\nabla \phi)^2<0$};
			\node[signal, signal to=east and west, draw, text width=10em, text centered, node distance=6.0cm, right of=nablaphi](CCy){$R - \frac{\Lambda}{m(m-1)} (g\odot g)=0$\\$\phi= \Phi$};
			\node[node distance=2.2cm, right of=CCy](point1){};
			\node[node distance=2.2cm, left of=CCy](point5){};
			\node[rectangle, draw, text width=8em, text centered, minimum height=24pt, below of=point1, node distance= 1.2cm](CS){$\CC^m_K \CS_\Phi$};
			\node[rectangle, draw,
			text centered,
			below of=point5, node distance= 1.2cm](nothing){not inflationary};
			\node[rectangle, draw, text centered, node distance=1.2cm, below of=nablaphi](Uphi){$U:= U_{\phi}$};
			\node[signal, signal to=east and west, draw, text centered, node distance=1.2cm, below of=Uphi](nablaU){$\nabla_i U_j =0$};
			\node[rectangle, draw, text centered, node distance=3.0cm, right of=nablaU](Krho){$K:= \frac{2\kappa\rho}{m(m-1)}$};
			\node[node distance=1.9cm, right of=nablaU](point4){};
			\node[signal, signal to=east and west, draw, text width=7em, text centered, node distance=3.5cm, right of=Krho](CCD){$\mathfrak{C}_{ijkh}=0$\\$(\nabla \phi)^2= -2\frac{\rho}{m}$\\$\zeta= \frac{2\kappa \rho}{m(m-1)}$};
			\node[rectangle, draw, text width=8em, text centered, minimum height=24pt, below of=point1, node distance= 3.7cm](CES){$\ESU^m_K \CES_{\rho,J}$};
			\node[rectangle, draw,
			text centered,
			below of=point5, node distance= 3.7cm](nothing2){not inflationary};
			\node[node distance=8.2cm, right of=nablaU](pointright){};
			\node[node distance=6.0cm, right of=nablaU](pointleft){};
			\node[node distance=2.4cm, below of=CCD](point3){};
			\node[node distance=3cm, right of=point3](point2){};
			\node[signal, signal to=east and west, draw, text width=5em, text centered, node distance=7.8cm, below of=nablaU](V2){$V= \frac{2\Lambda}{\kappa}$};
			\node[signal, signal to=east and west, draw, text width=5em, text centered, node distance=4.5cm,  below of=point4](zetaU1){$\zeta=0$};
			\node[signal, signal to=east and west, draw, text width=5em, text centered, node distance=11.2cm,  below of=point4](zetaU2){$\zeta=0$};
			\node[signal, signal to=east and west, draw, text width=11em, text centered, node distance=2.8cm,  below of=pointleft](gMMS){	$\nabla_i U_j - \frac{\nabla_i \phi'}{m\phi'} U_j - \xi g_{ij} = 0$ \\
				$\mathfrak{C}_{ijkh} = 0, \; {\xi \in J'}$\\
				$\phi' = -\sqrt{\Omega} |\zeta|^{\frac{m}{2}}, \; {\phi \in J}$ \\ 
				$\xi^2 + \zeta = \frac{2\Lambda + \kappa \Omega |\zeta|^m}{m(m-1)}$\\
			};
			\node[rectangle, draw, text width=8em, text centered, minimum height=24pt, below of=pointright, node distance= 4.5cm](gMMSy){$\MMS^m_{\Lambda,\Omega,J,J'}$};
			\node[signal, signal to=east and west, draw, text width=11em, text centered, node distance=6.1cm,  below of=pointleft](fMMS){$\nabla_i U_j - \frac{\nabla_i \phi'}{m\phi'} U_j - \xi g_{ij} = 0 $ \\
				$\mathfrak{Z}_{ijkh} = 0 , \; {\phi \in J}$\\
				$\xi^2 = \frac{\kappa \phi'^2 + 2\Lambda}{m(m-1)}, \; {\phi' \in J}$\\
				$\bm{\eta} + m \xi^2 = \frac{2\Lambda}{(m-1)}$
			};
			\node[rectangle, draw, text width=8em, text centered, minimum height=24pt, below of=pointright, node distance= 7.8cm](fMMSy){$\MMS^{m,0}_{\Lambda,J,J'}$};
			\node[rectangle, draw,
			text centered,
			below of=point5, node distance= 10.3cm](nothing3){not inflationary};
			\node[signal, signal to=east and west, draw, text width=12em, text centered, node distance=9.5cm,  below of=pointleft](gNKG){
				$\mathfrak{P}_{ij} = 0 , \mathfrak{D}_{ij} = 0 , \mathfrak{C}_{ijkh} = 0 $ \\
				$\zeta = \kappa \frac{\phi'^2 + V(\phi)}{m(m-1)} - \xi^2 $ \\
				${\xi = \Xi(\phi),} \; {\phi' = \Pi(\phi),} \; {\phi \in J}$
			};
			\node[rectangle, draw, text width=8em, text centered, minimum height=24pt, below of=pointright, node distance= 11.2cm](gNKGy){$\NKG^m_{V,\Pi,\Xi,J}$};
			\node[signal, signal to=east and west, draw, text width=12em, text centered, node distance=12.7cm,  below of=pointleft](fNKG){
				$\mathfrak{P}_{ij} = 0 , \mathfrak{D}_{ij} = 0 , \mathfrak{Z}_{ijkh} = 0$\\
				$\xi = \Xi(\phi), \; {\phi \in J}$	\\		
				$\phi' = -\frac{(m-1)}{\kappa} \del_\phi \Xi(\phi) $ \\
			};
			\node[rectangle, draw, text width=8em, text centered, minimum height=24pt, below of=pointright, node distance= 14.3cm](fNKGy){$\NKG^{m,0}_{V,\Xi,J}$};
			\node[rectangle, draw,
			text centered,
			below of=point5, node distance= 16.6cm](nothing4){not inflationary};
			\path[draw, -latex'] (first)--(nablaphi);
			\path[draw, -latex'] (nablaphi)--node {no}(CCy);
			\path[draw, -latex'] (CCy)--node [swap, pos=.8]{no}(nothing);
			\path[draw, -latex'] (CCy)--node [pos=1.1]{yes}(CS);
			\path[draw, -latex'] (nablaphi)--node {yes}(Uphi);
			\path[draw, -latex'] (Uphi)--node {}(nablaU);
			\path[draw, -latex'] (nablaU)--node {yes}(Krho);
			\path[draw, -latex'] (Krho)--node {}(CCD);
			\path[draw, -latex'] (CCD)--node [swap]{no}(nothing2);
			\path[draw, -latex'] (CCD)--node [pos=1.6]{yes}(CES);
			\path[draw, -latex'](nablaU) --node [pos=.2]{no}(V2);
			\path[draw, -latex'] (V2)--node [swap]{no}(zetaU2);
			\path[draw, -latex'] (V2)--node {yes}(zetaU1);
			\path[draw, -latex'] (zetaU1)--node {no}(gMMS);
			\path[draw, -latex'] (zetaU1)--node [swap]{yes}(fMMS);
			\path[draw, -latex'] (zetaU2)--node {no}(gNKG);
			\path[draw, -latex'] (zetaU2)--node [swap]{yes}(fNKG);
			\path[draw, -latex'] (gMMS)--node [pos=1.2]{yes}(gMMSy);
			\path[draw, -latex'] (fMMS)--node [pos=1.2]{yes}(fMMSy);
			\path[draw, -latex'] (gNKG)--node [pos=1]{yes}(gNKGy);
			\path[draw, -latex'] (fNKG)--node [pos=1]{yes}(fNKGy);
			\path[draw, -latex'] (gMMS)--node [pos=.8]{no}(nothing2);
			\path[draw, -latex'] (fMMS)--node [swap, pos=.8]{no}(nothing3);
			\path[draw, -latex'] (gNKG)--node {no}(nothing3);
			\path[draw, -latex'] (fNKG)--node [swap, pos=.8]{no}(nothing4);
			\end{tikzpicture}		
			\caption{IDEAL characterization of local isometry classes of regular inflationary spacetimes (Theorem~\ref{thm_infl_class}, Table~\ref{tab_infl_class}).\label{flowchart_infl_class}}	
		\end{center}
	\end{footnotesize} 
\end{figure}

\newpage

\section{Geometry of FLRW and inflationary spacetimes} \label{sec_FLRW_geom}

\begin{defn}\label{def_GRW}
Let $(F, g^F)$ be a $m$-dimensional Riemannian manifold, $m\ge 1$, $I
\subseteq \mathbb{R}$ an open interval with standard coordinate $t$ and
endowed with the usual reversed metric $-dt^2$ and $f \in
C^{\infty}(I)$, with $f>0$. A \emph{Generalized Robertson Walker (GRW)
spacetime} is a product manifold $M= I \times F$ endowed with the metric
$g$ defined as
\begin{equation} \label{GRW_metric}
	g = -\pi_I^* dt^2 + (f \circ \pi_I)^2 \pi_F^* g^F
\end{equation}
where $\pi_I$ and $\pi_F$ are respectively the projections on $I$ and
$F$. Furthermore $I$ is called the \emph{base}, $F$ the \emph{fiber} and
$f$ the \emph{warping function} (also \emph{scale factor}, in the
literature on cosmology). 
\end{defn} 
To simplify notation in the sequel, let us introduce the notation
$\tilde{T} = \pi^*_F T$ for any completely covariant tensor $T$ defined
on $F$.

The definition implies that around every point of $M = I\times F$, there
exists a coordinate system $(x^0,x^i)$ adapted to the product structure,
such that, denoting $t = x^0$,
\begin{equation}
	g_{ij} = -(dt)_i (dt)_j + f^2(t) g^F_{ij} ,
\end{equation}
where $g^F_{ij}$ depends only on the $x^i$ coordinates with $i>0$ and
$g^F_{ij} (\del_t)^i = 0$. The only obstacle to making the last
statement global on $M$ is that the $F$ factor may not admit a global
coordinate system.

\begin{defn} \label{def_FLRW}
A \emph{Friedmann-Lema\^itre-Robertson-Walker (FLRW) spacetime} is a
Lorentzian manifold $(M,g)$ that is a GRW spacetime
(Definition~\ref{def_GRW}) where the fiber $(F,g^F)$ is simply
connected, complete and has \emph{constant curvature} with sectional
curvature $\alpha$ (some constant), that is, the Riemann curvature
tensor $R^F_{ijkh}$ of $(F,g^F)$ has the form
\begin{equation} \label{const_curv}
	R^F_{ijkh} = \alpha (g^F_{ik}g^F_{hj}-g^F_{jk}g^F_{hi}) .
\end{equation}
\end{defn}
When $\dim M = 2$, only $\alpha = 0$ is possible, since any
$1$-dimensional $(F,g^F)$ is flat.

It is well known that any simply connected, complete Riemannian manifold
of constant curvature, meaning that its Riemann curvature tensor is of
the form~\eqref{const_curv}, is isometric to either a round sphere
($\alpha > 0$), flat Euclidean space ($\alpha = 0$), or a hyperbolic
space ($\alpha < 0$)~\cite[section 2.4]{Wolf2011}. If the
\emph{complete} and \emph{simply connected} hypotheses are dropped, then
a constant curvature Riemannian manifold is still locally isometric to
one of these model spaces.

Similarly, in the sequel, we will be interested in Lorentzian spacetimes
that are locally isometric (Definition~\ref{def_loc_isom}) to GRW or
FLRW models.

\subsection{Riemann curvature in GRW spacetimes} \label{sec_Riemann_GRW}

Below, we describe the Riemann curvature $R_{ijkh}$ in a GRW spacetime,
in terms of the curvature of $(F,g^F)$, the warping function $f$ and the
vector field $U_i = -(dt)_i$. For reference, let us denote the Riemann
tensor on the $(F,g^F)$ factor by $R^F_{ijkh}$, with $R^F_{ij} =
(g^F)^{kh} R^F_{ikjh}$ and $\mathcal{R}^F = (g^F)^{ij} R^F_{ij}$
denoting respectively the corresponding Ricci tensor and scalar. Recall
also the notation $\tilde{R}^F_{ijkh} = \pi_F^* R^{F}_{ijkh}$,
$\tilde{R}^F_{ij} = \pi_F^* R^F_{ij}$ and $\tilde{\mathcal{R}}^F =
\pi_F^* \mathcal{R}^F$.

Adapting the more general results on the covariant derivative on warped
products~\cite[Proposition 7.35]{ONeill1983}, the action of the
spacetime covariant derivative is determined by
\begin{align}
	\nabla_i (f U_j) = f' g_{ij} , \quad
	\nabla_i \tilde{X}_j = \widetilde{ \nabla_i X_j }
		- 2 \frac{f'}{f} U_{(i} \tilde{X}_{j)} ,
\end{align}
for any $X_j$ defined on $F$. Recalling the notation already used in the
introduction, for any $U_i$ we can define the temporal derivative $(-)'
:= U^i \nabla_i (-)$ and also
\begin{equation} \label{xi_eta}
	\xi := \frac{\nabla^i U_i}{m} , \quad
	\bm{\eta} := \xi' = U^j \nabla_j \xi .
\end{equation}
With our choice of $U$ on a GRW spacetime, we will be making repeated
use of the identifies
\begin{equation} \label{xi_eta_FLRW}
	\xi = \frac{f'}{f} , \quad
	\bm{\eta} = \frac{f''}{f} - \frac{f'^2}{f^2} .
\end{equation}
Geometrically $\xi$ is called the \emph{expansion} of the vector field
$U$, while its GRW value $f'/f$ is known as the \emph{Hubble rate} in
the literature on cosmology.

Next, adapting the more general result~\cite[Proposition
7.42]{ONeill1983} of how to write the Riemann tensor of a warped product
manifold in terms of the curvatures of its factors and the warping
function, it is possible to give the following general expression for the
Riemann tensor of GRW spacetimes:
\begin{align} \label{RiemannGRW_final}
	R_{ijkh}
\notag
	&= f^2 \tilde{R}^F_{ijkh}
		+ \left( g \odot \left[
				\frac{1}{2} \frac{f'^2}{f^2} g
				- \left(\frac{f''}{f}-\frac{f'^2}{f^2}\right) dt^2
			\right] \right)_{ijkh} \\
	&= f^2 \tilde{R}^F_{ijkh}
		+ \left( g \odot \left[
				\frac{\xi^2}{2} g
				- \bm{\eta} UU 
			\right] \right)_{ijkh} ,
\end{align}
where $(dt^2)_{ij} = (dt)_i (dt)_j$ and where we have used the product
notation~\eqref{kn-product}. When $m=1$, the tensors $g\odot UU$ and
$g\odot g$ are no longer linearly independent, in fact $g\odot UU =
-\frac{1}{2} g\odot g$. Moreover, the Riemann curvature for a
$1$-dimensional $(F,g^F)$ is always zero. Hence, in the special $m=1$
case we have the simplification
\begin{equation}
	R_{ijkh}
	= \frac{(\bm{\eta} + \xi^2)}{2} (g\odot g)_{ijkh}
	= \frac{f''}{f} \frac{1}{2} (g\odot g)_{ijkh} .
\end{equation}

As a consequence, using the identities
\begin{align}
	g^{kh} \tilde{R}^F_{ikjh}
		&= \frac{1}{f^2} \pi_F^* ((g^F)^{kh} R^F_{ikjh})
		= \frac{1}{f^2} \tilde{R}^F_{ij} , \\
	g^{ij} \tilde{R}^F_{ij}
		&= \frac{1}{f^2} \pi_F^* ((g^F)^{ij} R^F_{ij})
		= \frac{1}{f^2} \tilde{\mathcal{R}}^F ,
\end{align}
we get the following formulas for the Ricci tensor $R_{ij} = g^{kh}
R_{ikjh}$ and scalar $\mathcal{R} = g^{ij} R_{ij}$:
\begin{align}
	R_{ij}
\notag
		&= \tilde{R}^F_{ij}
			- (m-1)\left(\frac{f''}{f}-\frac{f'^2}{f^2}\right) U_i U_j 
			+ \left(\frac{f''}{f} + (m-1)\frac{f'^2}{f^2}\right) g_{ij} \\
		&= \tilde{R}^F_{ij}
			- (m-1) \bm{\eta} U_i U_j + (\bm{\eta}+ m\xi^2) g_{ij} , \\
	\mathcal{R}
\notag
		&= \frac{1}{f^2} \tilde{\mathcal{R}}^F
			+ 2m\frac{f''}{f} + m(m-1) \frac{f'^2}{f^2} \\
		&= \frac{1}{f^2} \tilde{\mathcal{R}}^F
			+ 2m \bm{\eta} + m(m+1) \xi^2 .
\end{align}
For completeness, we also compute the value of the scalar square of the
Ricci tensor:
\begin{align}
	\mathcal{B}
\notag
	&= \frac{\tilde{\mathcal{B}}^F}{f^4}
		+ 2\left(\frac{f''}{f}+(m-1)\frac{f'^2}{f^2}\right)
			\frac{\tilde{\mathcal{R}}^F}{f^2}
		+ m\left(\frac{f''}{f}+(m-1)\frac{f'^2}{f^2}\right)^2
		+ m^2 \frac{f''^2}{f^2} \\
	&= \frac{\tilde{\mathcal{B}}^F}{f^4}
		+ 2 (\bm{\eta} + m \xi^2) \frac{\tilde{\mathcal{R}}^F}{f^2}
		+ m (\bm{\eta} + m \xi^2)^2
		+ m^2 (\bm{\eta} + \xi^2)^2 ,
\end{align}
where $\mathcal{B}^F = (g^F)^{ik} (g^F)^{jh} R^F_{ij} R^F_{kh}$.

The above formulas motivate the following definitions, which can be used
to isolate the spatial curvature $R^F_{ijkh}$ from the knowledge of the
spacetime curvature $R_{ijkh}$ and of $U_i$.
\begin{defn} \label{curvature_cond}
Consider a Lorentzian manifold $(M,g)$ with a unit timelike vector field
$U$. Recall also the scalars $\xi$ and $\bm{\eta}$ scalars
from~\eqref{xi_eta}.

\begin{enumerate}
\item
We define the \emph{zero (spatial) curvature deviation (ZCD) tensor} as
\begin{equation} \label{zcdt}
	\mathfrak{Z}_{ijkh} := R_{ijkh}
	- \left( g \odot \left[ \frac{\xi^2}{2} g
		- \bm{\eta} UU \right] \right)_{ijkh} .
\end{equation}

\item
Provided $m>1$, we define the \emph{spatial curvature scalar} as
\begin{equation} \label{zeta}
	\zeta := \frac{\mathfrak{Z}_{ij}{}^{ij}}{m(m-1)}
	= \frac{\mathcal{R} - 2m\bm{\eta} - m(m+1)\xi^2}{m(m-1)}
\end{equation}
and if $m=1$, we set $\zeta = 0$.

\item
We define the \emph{constant (spatial) curvature deviation (CCD) tensor}
as
\begin{equation} \label{ccdt}
	\mathfrak{C}_{ijkh} := R_{ijkh}
	- \left( g \odot \left[ \frac{(\xi^2+\zeta)}{2} g
		- (\bm{\eta}-\zeta) UU \right] \right)_{ijkh} .
\end{equation}
\end{enumerate}
\end{defn}
On GRW spacetimes, the usefulness of these definitions lies in the
identities
\begin{gather}
\label{GRW_ZCDT}
	\mathfrak{Z}_{ijkh} = f^2 \tilde{R}^F_{ijkh} , \quad
	\zeta = \frac{1}{m(m-1)} \frac{\tilde{\mathcal{R}}^F}{f^2} , \\
\label{GRW_CCDT}
	\mathfrak{C}_{ijkh} =
		f^2 \left( \tilde{R}^F_{ijkh}
			- \frac{1}{m(m-1)} \frac{\tilde{\mathcal{R}}^F}{2}
				(\tilde{g}^F\odot \tilde{g}^F)_{ijkh} \right) .
\end{gather}

\subsection{Riemann curvature in FLRW spacetimes} \label{sec_Riemann}

Next, we specialize the main formulas obtained in the preceding section
from GRW to FLRW spacetimes (Definition~\ref{def_FLRW}), by making use
of their spatial curvature structure
\begin{equation}
	R^F_{ijkh}
	= \frac{\alpha}{2} (g^F\odot g^F)_{ijkh} , \quad
	R^F_{ij} = (m-1) \alpha g^F_{ij} , \quad
	\mathcal{R}^F = m(m-1) \alpha ,
\end{equation}
and of the identity
\begin{multline}
	f^2 \frac{1}{2} (\tilde{g}^F \odot \tilde{g}^F)_{ijkh}
	= \frac{1}{f^2} \frac{1}{2} \left((g+UU)\odot(g+UU)\right)_{ijkh} \\
	= \frac{1}{f^2} \left((g\odot UU)_{ijkh}
		+ \frac{1}{2} (g\odot g)_{ijkh}\right)
	= \frac{1}{f^2} \left( g \odot \left[ \frac{1}{2} g + UU \right] \right)_{ijkh} ,
\end{multline}
where we have recalled that $f^2 \tilde{g}^F = g + UU$. Recall also the
definitions of $U_i = (dt)_i$, the scalars $\xi$ and $\bm{\eta}$
from~\eqref{xi_eta}, and note the identity
\begin{equation}
	\zeta = \frac{\alpha}{f^2}
\end{equation}
for the spatial curvature scalar (Definition~\ref{curvature_cond}) when
$m>1$. When $m=1$, we always have $R^F_{ijkh} = 0$, so it is consistent
to take $\zeta = 0$, as we do.

Thus, for FLRW spacetimes of spatial sectional curvature $\alpha$, we have
\begin{align}
	R_{ijkh}
\label{Riem_FLRW}
	&= \left( g \odot \left[
			\frac{(\xi^2+\zeta)}{2} g
			- (\bm{\eta}-\zeta) UU 
		\right] \right)_{ijkh} , \\
	R_{ij}
\label{Ric_FLRW}
	&= -(m-1) (\bm{\eta} - \zeta) U_i U_j
		+ [(\bm{\eta} - \zeta) + m(\xi^2+\zeta)] g_{ij} , \\
	\mathcal{R}
\label{R_FLRW}
		&= m \left[ 2(\bm{\eta}-\zeta) + (m+1) (\xi^2+\zeta) \right] , \\
	\mathcal{B}
\label{B_FLRW}
	&= m [(\bm{\eta}-\zeta) + m(\xi^2+\zeta)]^2 + m^2 (\bm{\eta} + \xi^2)^2 ,
\end{align}
where we have also used $\mathcal{B}^F = m (m-1)^2 \alpha^2$. In the
special $m=1$ case, the above formulas simplify to
\begin{align}
	R_{ijkh} &= \frac{(\bm{\eta}+\xi^2)}{2} (g\odot g)_{ijkh} , \\
	R_{ij} &= (\bm{\eta} + \xi^2) g_{ij} , \\
	\mathcal{R} &= 2 (\bm{\eta}+\xi^2) , \\
	\mathcal{B} &= 2 (\bm{\eta}+\xi^2)^2 .
\end{align}

Because of the frequent appearance of the combinations $\bm{\eta} -
\zeta = \frac{f''}{f} - \frac{f'^2}{f^2} - \frac{\alpha}{f^2}$ and
$\xi^2 + \zeta = \frac{f'^2}{f^2} + \frac{\alpha}{f^2}$, in the sequel
we will need the identity
\begin{equation} \label{eq_fried_diff}
	(\xi^2 + \zeta)' = 2\xi (\bm{\eta} - \zeta)
	\quad \text{or} \quad
	\left(\frac{f'^2}{f^2}+\frac{\alpha}{f^2}\right)'
		= 2 \frac{f'}{f} \left(\frac{f''}{f}-\frac{f'^2}{f^2}-\frac{\alpha}{f^2}\right) .
\end{equation}

\subsection{Perfect fluid interpretation} \label{sec_perfect_fluid}

An arbitrary FLRW spacetime will in general not satisfy the vacuum
Einstein equations. But it could be interpreted, when $m>1$, as a
solution of Einstein equations with a perfect fluid stress energy tensor
\begin{equation}
	R_{ij} - \frac{1}{2} \mathcal{R} g_{ij} + \Lambda g_{ij}
		= \kappa T_{ij} = \kappa (\rho+p) U_i U_j + \kappa p g_{ij} ,
\end{equation}
where $\Lambda$ is the cosmological constant, $\rho$ is the \emph{energy
density} and $p$ is the pressure. When $m=3$, the coupling constant
usually has the value $\kappa = 8\pi G/c^4$, where $G$ is Newton's
constant and $c$ the speed of light. In other dimensions, there are at
least two conventions: either keeping the value of $\kappa$ the same, or
setting it to $\kappa = 2 \sigma_m G/c^4$, where $\sigma_m = 2
\pi^{\frac{m-1}{2}} / \Gamma(\frac{m-1}{2})$ is the area of the unit
$(m-1)$-sphere. We will simply keep it as an unspecified but fixed
constant $\kappa \ne 0$. The cosmological constant could of course be
shifted to $\Lambda \mapsto 0$ by the redefinitions $p \mapsto p -
\Lambda/\kappa$, $\rho \mapsto \rho + \Lambda/\kappa$. When $m=1$, the
fluid interpretation is no longer possible, simply because the Einstein
tensor $R_{ij} - \frac{1}{2} \mathcal{R} g_{ij}$ is identically zero in
two spacetime dimensions.

Defining $\mathcal{T} = g^{ij} T_{ij}$, an equivalent form of Einstein's
equations is
\begin{gather}
	R_{ij} = \kappa T_{ij} - \frac{\kappa}{m-1} \mathcal{T} g_{ij}
		= \kappa (\rho + p) U_i U_j + \kappa \frac{\rho-p}{m-1} g_{ij} .
\end{gather}
Hence, for FLRW spacetimes, these equations translate to
\begin{align}
\label{friedmann}
	\frac{f'^2}{f^2} + \frac{\alpha}{f^2}
		&= \frac{2}{m(m-1)} \kappa \rho ,
		&
	\kappa \rho &= \frac{m(m-1)}{2}
		\left( \frac{f'^2}{f^2} + \frac{\alpha}{f^2} \right) ,
		\\
\label{accel}
	\frac{f''}{f} - \frac{f'^2}{f^2} - \frac{\alpha}{f^2}
		&= -\frac{1}{m-1} \kappa (\rho + p) ,
		&
	\kappa p &= -(m-1) \left( \frac{f''}{f}
			- \frac{f'^2}{f^2} - \frac{\alpha}{f^2} \right) \\
\notag & & & \quad {}
		- \frac{m(m-1)}{2} \left( \frac{f'^2}{f^2} + \frac{\alpha}{f^2} \right) ,
\end{align}
On the top-left we have the \emph{Friedmann equation}, while on the
bottom-left we have the \emph{acceleration equation}. These equations
agree with the formulas previously obtained in~\cite{higher-flrw}, which
was one of the first to consider perfect fluid cosmologies in higher
spacetime dimensions. The Bianchi identity $\nabla^i (R_{ij} -
\frac{1}{2} \mathcal{R} g_{ij}) = 0$ implies the stress-energy
conservation $\nabla^iT_{ij} = 0$ condition, which translates to the
\emph{energy conservation} or \emph{continuity equation}
\begin{equation} \label{cont_eq}
	\rho' + m \frac{f'}{f} (\rho + p) = 0 .
\end{equation}

\subsection{Special FLRW classes} \label{sec_reg_FLRW}

Below, we list the forms of FLRW spacetimes (Definition~\ref{def_FLRW})
satisfying some special geometric conditions. Throughout this section,
consider an FLRW spacetime $(M,g)$, $\dim M = m+1 \ge 2$, with warping
function $f\colon I\to \mathbb{R}$ and spatial sectional curvature
$\alpha$. Whenever parameters are present, they must be chosen to
respect $f(t) > 0$ for all $t\in I$, even if not explicitly indicated,
as well as $\alpha=0$ when $m=1$.

\begin{lem} \label{lem_flat}
The complete list of possible triples $(m,\alpha,f(t))$ satisfying the
\emph{flat} (or \emph{Minkowski space}) condition, $R_{ijkh} = 0$,
consists of
\begin{equation*}
	\begin{cases}
		\frac{f'}{f}=0\colon &
			(m, 0, A) \quad (A>0) ;
		\\
		\frac{f'}{f}\ne 0\colon &
			\begin{cases}
				m=1\colon \frac{f''}{f}=0 &
					(1, 0, A(t-t_0)) \quad (A\ne 0) ; \\
				m>1\colon \frac{f'^2}{f^2} + \frac{\alpha}{f^2} = 0 &
					(m, \alpha, \pm\sqrt{-\alpha}(t-t_0)) \quad (\alpha < 0) .
			\end{cases}
	\end{cases}
\end{equation*}
\end{lem}
\begin{proof}
From Equation~\eqref{Riem_FLRW}, the necessary and sufficient conditions
are $\frac{f''}{f} = 0$ and $(\frac{f'^2}{f^2} + \frac{\alpha}{f^2}) =
0$, when $m>1$, or only $f''/f=0$, when $m=1$. It is easy to see that
the desired conclusion exhausts the solutions of these equations under
the constraint that $f(t) \ne 0$ everywhere.
\end{proof}

\begin{lem} \label{lem_const_curv}
The complete list of possible triples $(m,\alpha,f(t))$ satisfying the
\emph{constant curvature} (or \emph{(anti-)de~Sitter space}) condition,
$R_{ijkh} = \frac{K}{2} (g\odot g)_{ijkh}$, with sectional curvature $K
\ne 0$, consists of ($A$ constant)
\begin{equation*}
	\begin{cases}
		m=1\colon  \frac{f''}{f} = K &
			\begin{cases}
				K>0\colon &
					\left\{\begin{aligned}
					& \left(1,0,A\cosh(\sqrt{K}(t-t_0))\right) , \\
					& \left(1,0,A e^{\pm \sqrt{K}(t-t_0)}\right) \\
					& \qquad (A>0) ;
					\end{aligned}\right.
				\\
				K<0\colon &
					\left\{\begin{aligned}
					& \left(1,0,A\cos(\sqrt{-K}(t-t_0))\right) \\
					& \qquad (A> 0) ;
					\end{aligned}\right.
			\end{cases}
		\\
		m>1\colon 
			\begin{gathered} \textstyle
				\frac{f'^2}{f^2} + \frac{\alpha}{f^2} = K , \\ \textstyle
				\frac{f'}{f} \ne 0
			\end{gathered} &
			\begin{cases}
				K>0\colon &
					\left\{\begin{aligned}
					& \left(m, \alpha, \sqrt{\alpha/K} \cosh(\sqrt{K}(t-t_0))\right) , \\
					& \left(m, 0, A e^{\pm\sqrt{K}(t-t_0)}\right) \\
					& \qquad (\alpha>0, A>0) ;
					\end{aligned}\right.
				\\
				K<0\colon &
					\left\{\begin{aligned}
					& \left(m, \alpha, \sqrt{\alpha/K} \cos(\sqrt{-K}(t-t_0))\right) \\
					& \qquad (\alpha < 0) .
					\end{aligned}\right.
			\end{cases}
	\end{cases}
\end{equation*}
\end{lem}
\begin{proof}
Again, referring to Equation~\eqref{Riem_FLRW}, the necessary and
sufficient conditions are $\frac{f''}{f} - \frac{f'^2}{f^2} -
\frac{\alpha}{f^2} = 0$ and $\frac{f'^2}{f^2} + \frac{\alpha}{f^2} = K$,
when $m>1$, or only $f''/f = K$, when $m=1$. If $m>1$ and $f'/f = 0$, we
must have $K = \alpha/f^2 = 0$, which contradicts the $K\ne 0$
hypothesis. Otherwise, it is easy to see that the desired conclusion
exhausts the solutions of these equations under the constraint that
$f(t) \ne 0$ everywhere.
\end{proof}

\begin{lem} \label{lem_esu}
The complete list of possible triples $(m,\alpha,f(t))$ satisfying both
conditions $\mathcal{R}' = 0$ and $\mathcal{B}' = 0$, but not of
constant curvature, consists of ($A$, $K$ constant)
\begin{equation*}
	(m, K A^2, A) \quad (m>1, K \ne 0, A>0) .
\end{equation*}
\end{lem}
This is is the \emph{Einstein Static
Universe}~\cite[\textsection16.2]{Weinberg1972} with spatial sectional
curvature $K$, which solves the Einstein equation, $R_{ij} - \frac{1}{2}
\mathcal{R} g_{ij} + \Lambda g_{ij} = 0$, with the cosmological constant
$\Lambda=\frac{(m-1)(m-2)}{2}K$.
\begin{proof}
From Equations~\eqref{R_FLRW} and~\eqref{B_FLRW}, both $\mathcal{R}'=0$
and $\mathcal{B}'=0$ are third order equations in $f$. Eliminating
$f'''$ from both of them, we obtain the integrability condition
\begin{equation}
	m^2 (m-1)^2 \frac{f'}{f}
		\left(\frac{f''}{f}-\frac{f'^2}{f^2}-\frac{\alpha}{f^2}\right) = 0 .
\end{equation}
Obviously, it is trivial when $m=1$. This is not surprising, because
then $\mathcal{R}$ is the only independent curvature component and
$\mathcal{R}'=0$ already implies that the spacetime is of constant
curvature, which is excluded by the hypotheses.

Further, this integrability condition splits into the cases $f'/f = 0$
and $f'/f \ne 0$. In the latter, it implies $\frac{f''}{f} -
\frac{f'^2}{f^2} - \frac{\alpha}{f^2} = 0$ and $(\frac{f'^2}{f^2} +
\frac{\alpha}{f^2})' = 0$ (cf.~Equation~\eqref{eq_fried_diff}). But
these are precisely the necessary and sufficient conditions for the
spacetime to be of constant curvature (Lemma~\ref{lem_const_curv}),
which is excluded by our hypotheses. Thus, we are left with the only
possibility $f'/f=0$ and the desired conclusion clearly exhausts the
solutions of this equation.
\end{proof}

\begin{lem} \label{lem_csc}
The complete list of possible triples $(m,\alpha,f(t))$ satisfying the
constant scalar curvature condition $\mathcal{R}' = 0$, but with
$\mathcal{B}' \ne 0$, consists of
\begin{equation*}
	\begin{cases}
		\alpha = 0\colon &
			(m, 0, f) ~
			\left( m>1,
				~ \frac{f''}{f}-\frac{f'^2}{f^2}
					+ \frac{(m+1)}{2} (\frac{f'^2}{f^2} - K) = 0 ,
				~ \frac{f'^2}{f^2}-K \ne 0 \right) ;
		\\
		\alpha \ne 0\colon &
			(m, \alpha, f) ~
			\left( m>1,
				~\frac{f'^2}{f^2} + \frac{\alpha}{f^2}
					= K + \kappa \Omega \frac{|\alpha|^{\frac{m+1}{2}}}{f^{m+1}} ,
				~ \Omega \ne 0 \right) .
	\end{cases}
\end{equation*}
\end{lem}
These are FLRW spacetimes with cosmological constant $\Lambda =
\frac{m(m-1)}{2} K$ and \emph{radiation perfect fluid} of energy density
$\Omega_r = \frac{1}{\kappa} \left(\frac{f'^2}{f^2} + \frac{\alpha}{f^2}
- K\right)$, where $\Omega_r = C/f^{m+1}$ for some constant $C$. We
	refer to $\Omega_r$ as the radiation energy density because the term
	with the power law $1/f^{m+1}$ in the Friedmann equation
\begin{equation}
	\frac{f'^2}{f^2} + \frac{\alpha}{f^2} = K + \kappa \frac{C}{f^{m+1}} ,
\end{equation}
when considered by itself gives rise to the constitutive relation
$p_r(\rho) = \rho/m$, which is characteristic of radiation in thermal
equilibrium~\cite{nemiroff-patla}. If $\Omega_\alpha = \alpha/f^2$ is
the energy density due to spatial curvature, when it is nonzero, the
ratio $\Omega = \Omega_r/\Omega_\alpha$ defines our \emph{normalized
radiation density} constant $\Omega$.
\begin{proof}
If $f'/f = 0$, then $\mathcal{R} = m(m-1) \alpha /f^2$ and $\mathcal{B}
= m(m-1)^2 \alpha^2/f^4$. Hence $\mathcal{R}' = 0$ implies $\mathcal{B}
' = 0$. The same implication holds if $m=1$ (cf.~proof of
Lemma~\ref{lem_const_curv}). Therefore, by the $\mathcal{B}' \ne 0$
hypothesis, we can assume that $m>1$ and $f'/f \ne 0$.

From Equation~\eqref{R_FLRW}, the constant scalar curvature condition
$\mathcal{R}=m(m+1)K$ (with some constant $K$)
\begin{equation}
	2 \left(\frac{f''}{f} - \frac{f'^2}{f^2} - \frac{\alpha}{f^2}\right)
		+ (m+1)\left(\frac{f'^2}{f^2} + \frac{\alpha}{f^2}\right)
		= (m+1) K ,
\end{equation}
after multiplying both sides by the integrating factor $(f'/f) f^{m+1}$
and using identity~\eqref{eq_fried_diff}, it is equivalent to
\begin{equation}
	f^{m+1} \left( \frac{f'^2}{f^2} + \frac{\alpha}{f^2} - K \right)
	= \kappa C ,
\end{equation}
for some constant $C$. If $C=0$, we are back to the case of constant
curvature (Lemma~\ref{lem_const_curv}), which is excluded by the
$\mathcal{B}' \ne 0$ hypothesis. When $\alpha \ne 0$, we can normalize
this constant as $C = \Omega |\alpha|^{\frac{m+1}{2}}$, with some
$\Omega \ne 0$. Thus, the desired conclusion clearly consists of the
necessary and sufficient conditions for $\mathcal{R}' = 0$ and
$\mathcal{B}' \ne 0$ to hold.
\end{proof}

\begin{lem} \label{lem_flat_FLRW}
For any triple $(m,\alpha,f(t))$ for which $\alpha = 0$, $(f'^2/f^2)'
\ne 0$ and $(\nabla\mathcal{R})^2 < 0$, there is a unique smooth
function $P\colon J \to \mathbb{R}$, where $J = \frac{f'^2}{f^2}(I)$, $I\subseteq\mathbb{R}$ and
\begin{equation}
	\frac{f''}{f} - \frac{f'^2}{f^2} + \frac{m}{2} \frac{f'^2}{f^2}
		= -\kappa P\left((f'/f)^2\right) .
\end{equation}
The function $P(u)$ will also satisfy the following condition for each
$u\in J$:
\begin{equation}
	P(u) \left[\kappa \del_u P(u) - \frac{1}{2}\right] \ne 0 . 
\end{equation}
\end{lem}
We will call $P$ the \emph{normalized pressure function} because, when
$m>1$, the spacetime admits a perfect fluid interpretation
(Section~\ref{sec_perfect_fluid}) with energy density $\kappa \rho(t) =
\frac{m(m-1)}{2} (f'/f)^2$, pressure $p(t) = (m-1)
P\left((f'/f)^2\right)$, which admits the constitutive relation $p =
p(\rho)$, where
\begin{equation}
	p(\rho) = (m-1)P\left(\frac{2}{m(m-1)} \kappa \rho\right) .
\end{equation}
When $m=1$, the triviality of Einstein equations doesn't allow such
an interpretation, so without loss of generality the function $P$ simply
determines the differential equation satisfied by $f$.
\begin{proof}
Under our hypotheses, the existence of a unique function $P(u)$ is an
elementary consequence of the implicit function theorem. If $P(u) = 0$,
then we are back to the case of flat or constant curvature spacetime
(Lemmas~\ref{lem_flat}, \ref{lem_const_curv}), while $P(u) =
\frac{1}{2\kappa} [u - (m+1)K]$ brings us back to the $\mathcal{R}'=0$
case (Lemma~\ref{lem_csc}), both of which contradict the $(\nabla
\mathcal{R})^2 < 0$ hypothesis. For any other value of $P(u)$, we have
$\nabla \mathcal{R} \ne 0$, which can then only be timelike.
\end{proof}

\begin{lem} \label{lem_gen_FLRW}
For any triple $(m,\alpha,f(t))$ for which we have $\alpha \ne 0$, $f'/f
\ne 0$ and $(\nabla\mathcal{R})^2 < 0$, there is a unique smooth
function $E\colon J \to \mathbb{R}$, where $J = \frac{\alpha}{f^2}(I)$, $I\subseteq\mathbb{R}$ and
\begin{equation}
	\frac{f'^2}{f^2} + \frac{\alpha}{f^2}
		= \kappa E(\alpha/f^2) .
\end{equation}
The function $E(u)$ will also satisfy the following conditions for each
$u\in J$:
\begin{equation} \label{eq_gen_FLRW_ineq}
	\kappa E(u) - u > 0 , \quad
	\del_u \left[u\del_u E(u) - \frac{(m+1)}{2} E(u)\right] \ne 0 .
\end{equation}
\end{lem}
We will call $E$ the \emph{normalized energy function} because, when
$m>1$, the spacetime admits a perfect fluid interpretation
(Section~\ref{sec_perfect_fluid}) with energy density $\rho(t) =
\frac{m(m-1)}{2} E(\alpha/f^2)$ and pressure $p(t) = -(f\rho')/(mf') -
\rho$ given by the continuity equation~\eqref{cont_eq}. When $m=1$, the
triviality of the Einstein equations doesn't allow such an
interpretation, so without loss of generality the function $E$ simply
determines the differential equation satisfied by $f$.
\begin{proof}
Under our hypotheses, the existence of a unique function $E(u)$ is an
elementary consequence of the implicit function theorem. Since $f'^2/f^2
> 0$, we must also have $\kappa E(u) - u > 0$. Finally, we want to make
sure that $\kappa E(u) \ne K + \kappa \Omega u^{\frac{m+1}{2}}$, which
would imply $\mathcal{R}' = 0$ (Lemma~\ref{lem_csc}), contrary to our
hypothesis that $(\nabla \mathcal{R})^2 < 0$. With $K$ and $\Omega$
arbitrary, these right-hand-sides precisely exhaust the solutions of the
equation $u\del_u (u\del_u - \frac{m+1}{2}) E(u) = 0$. Thus, the second
inequality in~\eqref{eq_gen_FLRW_ineq} is sufficient to ensure that
$\nabla\mathcal{R} \ne 0$, which can then only be timelike.
\end{proof}

\subsection{Scalar field} \label{sec_scalar}

In this section, we will be interested in the geometry of Lorentzian
spacetimes that are endowed with a scalar field and satisfying the
coupled Einstein equations. To make non-trivial use of Einstein
equations, throughout this section we will assume that the spacetime
dimension is $m+1>2$. This information will later be used in
Section~\ref{infl_isom_class} to classify the local isometry classes
(Definition~\ref{def_loc_isom}) of such spacetimes.

\begin{defn} \label{def_infl}
We call a spacetime with scalar $(M,g,\phi)$, with $\dim M = m+1 > 2$, an
\emph{inflationary spacetime} when $(M,g)$ can be put in FLRW
form~\eqref{GRW_metric}, $(M,g) \cong (I\times F, -dt^2 + f^2 g^F)$ such
that $\phi = \phi(t)$ is only a function of the $t$-coordinate, and for
some constant $\Lambda$ and smooth function $V(\phi)$ the coupled
Einstein-Klein-Gordon equations are satisfied
\begin{gather}
\label{eq_KG}
	\nabla^i \nabla_i \phi - \frac{1}{2}\del_\phi V(\phi) = 0 , \\
\label{eq_EEKG}
	R_{ij} - \frac{1}{2} g_{ij} \mathcal{R} + \Lambda g_{ij} = \kappa T_{ij} , \\
\notag
	\text{where} \quad
	T_{ij} = (\nabla_i \phi) (\nabla_j \phi)
		- \frac{1}{2} g_{ij} [(\nabla \phi)^2 + V(\phi)] .
\end{gather}
\end{defn}
Equation~\eqref{eq_KG} is in general the \emph{nonlinear Klein-Gordon
equation} with $V(\phi)$ the self-coupling potential, though in the
special case that the potential is a quadratic polynomial it becomes
linear. It is easy to see that we can set $\Lambda \mapsto 0$ by the
redefinition $V(\phi) \mapsto V(\phi) + \frac{2}{\kappa} \Lambda$. We
will adopt this convention from now on.

On an FLRW background, when $\phi = \phi(t)$, the stress energy tensor
and the wave operator are given by
\begin{gather}
	T_{ij} = \phi'^2 U_i U_j + \frac{1}{2}[\phi'^2 - V(\phi)] g_{ij} , \\
	\nabla^i \nabla_i \phi = -\phi'' - m \frac{f'}{f} \phi' .
\end{gather}
Hence, the coupled Einstein-Klein-Gordon equations reduce to the system of
ODEs
\begin{align}
\label{eq_fried}
	\frac{f'^2}{f^2} + \frac{\alpha}{f^2}
		&= \kappa \frac{\phi'^2 + V(\phi)}{m(m-1)} , \\
\label{eq_accel}
	\frac{f''}{f} - \frac{f'^2}{f^2} - \frac{\alpha}{f^2}
		&= -\kappa\frac{\phi'^2}{(m-1)} , \\
\label{eq_NLKG}
	\phi'' + \frac{1}{2} \del_\phi V(\phi)
		&= - m\frac{f'}{f} \phi' ,
\end{align}
which we will refer to as the \emph{Friedmann
equation}~\eqref{eq_fried}, the \emph{(Einstein) acceleration
equation}~\eqref{eq_accel}, and the \emph{nonlinear Klein-Gordon
equation}. When $\phi' \ne 0$, the nonlinear Klein-Gordon equation is
not independent from the other two and follows from the continuity
equation~\eqref{cont_eq} applied to this situation. Note that the
potential $V(\phi)$ can be isolated from the following combination of
the Friedmann and acceleration equations:
\begin{align}
\notag
	\kappa \frac{V(\phi)}{(m-1)}
	&= \left(\frac{f''}{f} - \frac{f'^2}{f^2} - \frac{\alpha}{f^2}\right)
		+ m \left(\frac{f'^2}{f^2} + \frac{\alpha}{f^2}\right)
	\\
	&= \frac{f''}{f} + (m-1) \left(\frac{f'^2}{f^2} + \frac{\alpha}{f^2}\right) .
\end{align}

While we will eventually give a characterization of local isometry
classes of inflationary spacetimes with a specific scalar potential
$V(\phi)$, it is an interesting question how to recognize when an FLRW
spacetime can be interpreted as part of a solution to an
Einstein-Klein-Gordon system with \emph{some} potential $V(\phi)$. This
is a coarser version of the question that asks for a Rainich-type
characterization with a \emph{specific} potential $V(\phi)$. The latter
finer question was answered in Theorem~4 of~\cite{krongos-torre}, on
which we base the following considerations.

Our starting point are the equations
\begin{align}
	-\kappa \frac{\phi'^2}{(m-1)}
		&= \frac{f''}{f} - \frac{f'^2}{f^2} - \frac{\alpha}{f^2} , \\
\label{eq_V}
	\kappa \frac{V(\phi)}{(m-1)}
		&= \left(\frac{f''}{f} - \frac{f'^2}{f^2} - \frac{\alpha}{f^2}\right)
			+ m \left(\frac{f'^2}{f^2} + \frac{\alpha}{f^2}\right) .
\end{align}
To answer our question, we will be happy with some reasonable conditions
on a given $(\alpha,f)$ for the existence of $\phi(t)$ and $V(\phi)$
such that the above equations are satisfied. Supposing that the
potential $V(\phi)$ has a smooth inverse, $V(\phi) = u \iff \phi =
W(u)$, we have the relation $(V(\phi))'/\phi' =
1/W'(V(\phi))$, which is of course consistent only if both expressions
remain both finite and non-zero. On the other hand, knowing $W'(u)$, we
can recover $W$ up to the ambiguity $W(u) \mapsto W(u) + \phi_0$, which
determines $V$ up to the ambiguity $V(\phi) \mapsto V(\phi-\phi_0)$.
Thus, under the hypotheses
\begin{gather}
\label{EEKG_ineq1}
	-\frac{1}{\kappa} \left(\frac{f''}{f} - \frac{f'^2}{f^2} - \frac{\alpha}{f^2}\right) > 0 ,
	\\
\label{EEKG_ineq2}
	\left[
		\left(\frac{f''}{f} - \frac{f'^2}{f^2} - \frac{\alpha}{f^2}\right)
		+ m \left(\frac{f'^2}{f^2} + \frac{\alpha}{f^2}\right)
	\right]' \ne 0 ,
\end{gather}
using the last left-hand-side as the independent variable in an
application of the implicit function theorem, we define functions
$W'$ by the formula
\begin{multline}
	\frac{\frac{(m-1)}{\kappa} \left[
		\left(\frac{f''}{f} - \frac{f'^2}{f^2} - \frac{\alpha}{f^2}\right)
		+ m \left(\frac{f'^2}{f^2} + \frac{\alpha}{f^2}\right)
	\right]'}{\pm\sqrt{-\frac{(m-1)}{\kappa}
		\left(\frac{f''}{f} - \frac{f'^2}{f^2} - \frac{\alpha}{f^2}\right)}}
	\\
	= \frac{1}{W'\left(\frac{(m-1)}{\kappa} \left[ \left(\frac{f''}{f} - \frac{f'^2}{f^2} - \frac{\alpha}{f^2}\right)
		+ m \left(\frac{f'^2}{f^2} + \frac{\alpha}{f^2}\right) \right] \right)} ,
\end{multline}
which fixes $W$ uniquely up to the ambiguity, $W(u) \mapsto \pm W(u) +
\phi_0$. Hence, we can let $V(\phi) = W^{-1}(\phi)$ and
\begin{equation}
	\phi(t) = W\left(\frac{(m-1)}{\kappa} \left[ \left(\frac{f''}{f} - \frac{f'^2}{f^2} - \frac{\alpha}{f^2}\right)
		+ m \left(\frac{f'^2}{f^2} + \frac{\alpha}{f^2}\right) \right]\right) ,
\end{equation}
which are unique up to the ambiguity $V(\phi) \mapsto V(\pm
[\phi-\phi_0])$ and $\phi(t) \mapsto \pm [\phi(t) - \phi_0]$. With these
definitions for $\phi(t)$ and $V(t)$, $(\alpha,f)$ will satisfy the
desired coupled Einstein-Klein-Gordon equations. Thus, any FLRW
spacetime satisfying the inequalities~\eqref{EEKG_ineq1}
and~\eqref{EEKG_ineq2} can be thought of as part of a solution of the
Einstein-Klein-Gordon equations with some non-constant potential. On the
other hand, the conditions on $\alpha$ and $f$ to be part of a solution
of Einstein-Klein-Gordon equations with a constant potential are
considered in Lemma~\ref{lem_mms}.

\subsection{Special inflationary classes} \label{sec_reg_infl}

Below, we list the forms of inflationary spacetimes
(Definition~\ref{def_infl}) satisfying some special geometric
conditions. Throughout this section, consider an inflationary spacetime
$(M,g,\phi)$, $\dim M = m+1 > 2$, with scalar field $\phi \colon I \to
\mathbb{R}$, warping function $f\colon I\to \mathbb{R}$ and spatial
sectional curvature $\alpha$. Whenever parameters are present, they must
be chosen to respect $f(t) > 0$ for all $t\in I$, even if not explicitly
indicated.

\begin{lem} \label{lem_flat_infl}
The complete list of possible quadruples $(m,\alpha,f(t),\phi(t))$
satisfying the constant scalar condition, $\phi(t) = \Phi$, as well as
$f'/f \ne 0$, consists of $(m,\alpha,f,\Phi)$ with $(m,\alpha,f)$
satisfying the constant curvature condition, $R_{ijkh} = \frac{K}{2}
(g\odot g)_{ijkh}$, with some spacetime sectional curvature constant
$K$. The Einstein-Klein-Gordon equations are satisfied with the choice
$V(\phi) = \frac{2}{\kappa} \Lambda$, where the cosmological constant
$\Lambda = \frac{m(m-1)}{2} K$.
\end{lem}
\begin{proof}
Since $\phi' = 0$, the Einstein-Klein-Gordon equations reduce to $R_{ij}
- \frac{1}{2} \mathcal{R} g_{ij} = -\Lambda g_{ij}$, or $R_{ij} = m K
g_{ij}$, with $K = \frac{2}{m(m-1)} \Lambda$, which together with the
FLRW property is precisely the necessary and sufficient to be of
constant curvature.
\end{proof}

Further on, in several cases, we will require the condition $f'/f \ne
0$. So first, we explore the special case $f'/f = 0$, of static
backgrounds. We know from Lemma~\ref{lem_esu} that the only static FLRW
backgrounds are flat or Einstein static universes, with the flat case
already covered by Lemma~\ref{lem_flat_infl}. What is special about this
case is that the energy $\frac{1}{2} (\phi'^2 + V(\phi))$ of the scalar
field is conserved. It turns out that the converse is also true and it
is only consistent with $V(\phi)$ being constant.

\begin{lem} \label{lem_ces} 
The complete list of possible quadruples $(m,\alpha,f(t),\phi(t))$
satisfying the constant energy condition $\frac{1}{2} (\phi'^2 +
V(\phi)) = \rho$, with some constant $\rho$, but with $(m,\alpha,f(t))$
not of constant curvature, consists of
\begin{equation}
	(m, K A^2, A, \pm\sqrt{2\rho/m} (t-t_0)) \quad
	\left( A > 0 , ~~ \rho > 0 \right) .
\end{equation}
\end{lem}
\begin{proof}
We can presume that $\phi' \ne 0$, since otherwise the spacetime is of
constant curvature (Lemma~\ref{lem_flat_infl}). The Friedmann
equation~\eqref{eq_fried} reduces to $f'^2/f^2+\alpha/f^2 = \kappa
\rho$. Using the identity~\eqref{eq_fried_diff} and the acceleration
equation~\eqref{eq_accel}, we conclude that $f'/f = 0$. Plugging this
conclusion back into the Friedmann and acceleration equation, we find
that each of $K = \alpha/f^2$, $\frac{2}{\kappa} \Lambda = V(\phi)$ and
$\phi'^2$ must be individually constant, with $K$ interpreted as the
spatial sectional curvature and $\Lambda$ the cosmological constant. If
we take $\rho$ as an independent constant, the rest are given by $K =
\frac{2}{m(m-1)} \kappa \rho$, $\phi'^2 = \frac{(m-1)}{\kappa} K =
\frac{2}{m} \rho$ and $\Lambda = \frac{(m-1)}{m} \kappa \rho$.
\end{proof}

Whenever the scalar potential $V(\phi)$ is a constant, the Klein-Gordon
equation is just the wave equation $\nabla^i \nabla_i \phi = 0$, which
we also call the \emph{massless minimally-coupled Klein-Gordon
equation}.

\begin{lem} \label{lem_mms}
The complete list of possible quadruples $(m,\alpha,f(t),\phi(t))$ with
$V(\phi) = \frac{2}{\kappa} \Lambda$ a constant, where the scalar field
is not constant nor of constant energy, consists of
\begin{equation*}
	\begin{cases}
		\alpha = 0\colon & (m,0,f,\phi) ~~
			\left( \begin{gathered} \textstyle
				\left(\frac{f''}{f}-\frac{f'^2}{f^2}\right)
					+ m\frac{f'^2}{f^2} = \frac{2\Lambda}{(m-1)} , \\ \textstyle
				\frac{f'^2}{f^2} = \frac{\kappa\phi'^2+2\Lambda}{m(m-1)}
			\end{gathered} \right)
		\\
		\alpha \ne 0\colon & (m,\alpha,f,\phi) ~~
			\left( \begin{gathered} \textstyle
				\frac{f'^2}{f^2}+\frac{\alpha}{f^2} =
					\frac{2\Lambda}{(m-1)}
					+ \frac{\kappa}{m(m-1)} \Omega \frac{|\alpha|^m}{f^{2m}} , \\ \textstyle
					\phi' = \pm\sqrt{\Omega} \frac{|\alpha|^{\frac{m}{2}}}{f^m} , ~~
					\Omega > 0
			\end{gathered} \right)
	\end{cases}
\end{equation*}
\end{lem}
\begin{proof}
Recall from~\eqref{eq_V} that a constant potential $V(\phi) =
\frac{2}{\kappa} \Lambda$ implies the equation
\begin{equation} \label{eq_EEKG_massless}
	\left(\frac{f''}{f} - \frac{f'^2}{f^2} - \frac{\alpha}{f^2}\right)
		+ m \left(\frac{f'^2}{f^2} + \frac{\alpha}{f^2}\right)
	= 2\frac{\Lambda}{(m-1)} ,
\end{equation}
which is also supplemented by the Friedmann equation~\eqref{eq_fried}
\begin{equation}
	\frac{f'^2}{f^2} + \frac{\alpha}{f^2}
		= \frac{\kappa \phi'^2 + 2\Lambda}{m(m-1)}
\end{equation}
is clearly equivalent to the Einstein equations with a massless
minimally-coupled scalar field stress energy tensor and, because of our
hypothesis that $\phi' \ne 0$ and the comments below
Equation~\eqref{eq_NLKG}, which are equivalent to the full coupled
Einstein-Klein-Gordon system. Setting $\alpha = 0$ completes the proof of the first part of the lemma.

The hypothesis of non-constant energy and Lemma~\ref{lem_ces} imply that
$f'/f\neq 0$. Thus, we obtain the following equivalent form
of~\eqref{eq_EEKG_massless} after multiplying it by the integrating
factor $2(f'/f)f^{2m}$:
\begin{equation}
	f^{2m} \left(\frac{f'^2}{f^2} + \frac{\alpha}{f^2}
		- \frac{2}{m(m-1)} \Lambda \right)
	= \frac{\kappa}{m(m-1)} C ,
\end{equation}
for some constant $C$. When $\alpha \ne 0$, we can normalize $C$ by
a power of $|\alpha|$ to get
\begin{equation} \label{eq_fried_mms}
	\frac{f'^2}{f^2} + \frac{\alpha}{f^2}
	= \frac{2\Lambda + \kappa \Omega \frac{|\alpha|^m}{f^{2m}}}{m(m-1)} ,
\end{equation}
with another constant $\Omega$. Provided that $\Omega > 0$, we can
determine $\phi(t)$ by the equation $\phi' = \pm \sqrt{\Omega}
\frac{|\alpha|^{\frac{m}{2}}}{f^m}$, which is equivalent to the massless
minimally-coupled Klein-Gordon equation
\begin{equation}
	\frac{1}{f^m} \left( f^m \phi' \right)'
		= \phi'' + m\frac{f'}{f} \phi' = 0 .
\end{equation}
With the above expression for $\phi'$, plugging it into the Friedmann
equation gives exactly Equation~\eqref{eq_fried_mms}. This observation
completes the proof of the second part of the lemma.
\end{proof}

Next, we will transform the Einstein-Klein-Gordon
equations~\eqref{eq_fried}, \eqref{eq_accel} and~\eqref{eq_NLKG} under
the hypothesis that $\phi'\ne 0$ everywhere. If we use the Friedmann
equation to eliminate $\alpha/f^2$ from the acceleration equation, while
also multiplying the Klein-Gordon equation by $\phi'$ and adding to it a
multiple of the acceleration equation, they can be equivalently
expressed as
\begin{align}
\label{eq_fried2}
	\kappa\frac{\phi'^2+V(\phi)}{m(m-1)} - \frac{f'^2}{f^2}
		&= \frac{\alpha}{f^2} , \\
\label{eq_accel2}
	\left(\frac{f'}{f}\right)' + \kappa \frac{\phi'^2}{(m-1)}
		&= \left( \kappa\frac{\phi'^2+V(\phi)}{m(m-1)} -  \frac{f'^2}{f^2} \right) , \\
\label{eq_NLKG2}
	\left( \kappa\frac{\phi'^2+V(\phi)}{m(m-1)} - \frac{f'^2}{f^2} \right)'
		&= -2\frac{f'}{f} \left( \kappa\frac{\phi'^2+V(\phi)}{m(m-1)} - \frac{f'^2}{f^2} \right) .
\end{align}
The equations~\eqref{eq_accel2} and~\eqref{eq_NLKG2} are second order,
while~\eqref{eq_fried2} is first order. To see that there are no
integrability conditions, note that differentiating the
first order equation gives the identity
\begin{multline} \label{eq_alpha_const}
	\left[ f^2 \left( \kappa
		\frac{\phi'^2+V(\phi)}{m(m-1)} - \frac{f'^2}{f^2} \right) \right]'
	\\
	= f^2 \left[ \left( \kappa
				\frac{\phi'^2+V(\phi)}{m(m-1)} - \frac{f'^2}{f^2} \right)'
			+ 2\frac{f'}{f} \left( \kappa
				\frac{\phi'^2+V(\phi)}{m(m-1)} - \frac{f'^2}{f^2} \right) \right] ,
\end{multline}
where the right-hand-side is clearly proportional to~\eqref{eq_NLKG2}.

Since we are assuming that $\phi'\ne 0$, we can use $\phi$ as the
independent variable and convert all $t$-derivatives as $(-)' = \phi'
\del_\phi(-)$. Denoting $\pi = \phi'$ and $\xi = f'/f$, we get the
equations
\begin{align}
\label{eq_fried3}
	f^2 \left( \kappa\frac{\pi^2 + V(\phi)}{m(m-1)} - \xi^2 \right)
		&= \alpha , \\
\label{eq_accel3}
	\pi \left[\del_\phi \xi + \kappa \frac{\pi}{(m-1)}\right]
		&= \left( \kappa\frac{\pi^2 + V(\phi)}{m(m-1)} - \xi^2 \right) , \\
\label{eq_NLKG3}
	\del_\phi \left( \kappa\frac{\pi^2 + V(\phi)}{m(m-1)} - \xi^2 \right)
		&= -2\frac{\xi}{\pi}
			\left( \kappa\frac{\pi^2 + V(\phi)}{m(m-1)} - \xi^2 \right) ,
\end{align}
where $\xi$, $\pi$ and $f$ are now all considered as functions of
$\phi$. With fixed $V(\phi)$, the system~\eqref{eq_accel3},
\eqref{eq_NLKG3} closes in the $(\pi,\xi)$ variables, with the symmetry
$(\pi,\xi) \mapsto (-\pi,-\xi)$ corresponding to the coordinate
transformation $t \mapsto -t$, and can be solved for the highest
derivatives $\del_\phi \xi$ and $\del_\phi \pi$ (always assuming that
$\pi \ne 0$). In the notation of~\eqref{eq_pi_xi}, we can use the
short-hand $\mathfrak{G}_V(\pi,\xi) = 0$ for this system. Hence the
space of solutions $\xi = \Xi(\phi)$, $\pi = \Pi(\phi)$, will be
two-dimensional. We will always leave these parameters implicit in the
choice of the solution $(\Pi(\phi),\Xi(\phi))$. With $(\Pi,\Xi)$ fixed,
the equations $\phi' = \Pi(\phi)$, $f'/f = \Xi(\phi)$ and $\alpha = f^2
[\kappa\frac{\Pi^2(\phi) + V(\phi)}{m(m-1)} - \Xi^2(\phi)]$ have a
two-dimensional family of solutions, parametrized essentially by the
transformations
\begin{equation}\label{eq_isom_At0}
	(\alpha, f(t), \phi(t)) \mapsto
		(A^2\alpha, A f(t-t_0), \phi(t_0-t_0)),
\end{equation}
which are the isometries preserving FLRW form
(Proposition~\ref{prop_isometries}). So the parameters determining
$(\alpha,f,\phi)$ that are invariant under these transformations are
essentially exhausted by the choice of $V(\phi)$ and the solution
$(\Pi,\Xi)$. We summarize as follows.

\begin{lem} \label{lem_gen_NLKG}
For any quadruple $(m,\alpha,f(t),\phi(t))$ for which $\alpha \ne 0$,
$f'/f \ne 0$ and $(\nabla \phi)^2 < 0$, there is a unique smooth
function $(\Pi,\Xi) \colon J \to \mathbb{R}^2$, where $J = \phi(I)$
and
\begin{equation}
	\phi' = \Pi(\phi) , \quad
	\frac{f'}{f} = \Xi(\phi) , \quad
	\frac{\alpha}{f^2} = \kappa \frac{\Pi^2(\phi) + V(\phi)}{m(m-1)}
		- \Xi^2(\phi) .
\end{equation}
For each $u\in J$, these functions will also satisfy $\Pi(u) \ne 0$,
$\Xi(u) \ne 0$ and $\kappa \frac{\Pi^2(u) + V(u)}{m(m-1)} \ne \Xi^2(u)$,
they will satisfy $\mathfrak{G}_V(\Pi,\Xi) = 0$, in the notation
of~\eqref{eq_pi_xi}.
\end{lem}

When $\alpha = 0$, the above discussion can be greatly simplified. The
Einstein-Klein-Gordon system reduces to the following equivalent forms,
using the same notation as above and always supposing that $\pi \ne 0$
everywhere:
\begin{equation}
	\left\{\begin{aligned}
		\kappa \frac{\pi+V(\phi)}{m(m-1)} - \xi^2 &= 0 \\
		\pi \left[\del_\phi \xi + \kappa\frac{\pi}{(m-1)}\right] &= 0
	\end{aligned}\right.
	\iff
	\left\{\begin{aligned}
		(\del_\phi\xi)^2 &= \kappa \frac{m(m-1) \xi^2 - \kappa V(\phi)}{(m-1)^2} \\
		\pi &= - \frac{(m-1)}{\kappa} \del_\phi \xi
	\end{aligned}\right.
\end{equation}
In a way, this simplification comes from eliminating $\pi = \phi'$ from
the equations. In the notation of~\eqref{eq_xi_flat}, we use the
short-hand $\mathfrak{H}_V(\xi) = 0$ for the equation satisfied by
$\xi(\phi)$, which retains the symmetry $\xi \mapsto -\xi$. With
$V(\phi)$ fixed, under the hypothesis $\del_\phi \Xi \ne 0$, this
equation will have a one-dimensional family of solutions $\xi =
\Xi(\phi)$. We will always leave the corresponding parameter implicit in
the choice of the solution $\Xi(\phi)$. With $\Xi$ fixed, the equations
$\phi' = -\frac{(m-1)}{\kappa} \del_\phi \Xi(\phi)$, $f'/f = \Xi(\phi)$
have a two-dimensional family of solutions, again parametrized by the
transformations~\eqref{eq_isom_At0}. So the parameters determining
$(f,\phi)$ that are invariant under these transformations are
essentially exhausted by the choice of $V(\phi)$ and the solution $\Xi$.
We summarize as follows.

\begin{lem} \label{lem_flat_NLKG}
For any quadruple $(m,\alpha,f(t),\phi(t))$ for which $\alpha = 0$,
$f'/f \ne 0$ and $(\nabla \phi)^2 < 0$, there is a unique smooth
function $\Xi \colon J \to \mathbb{R}$, where $J = \phi(I)$
and
\begin{equation}
	\phi' = -\frac{(m-1)}{\kappa} \del_\phi\Xi(\phi) , \quad
	\frac{f'}{f} = \Xi(\phi) .
\end{equation}
For each $u\in J$, these functions will also satisfy $\Xi(u) \ne 0$,
$\del_\phi\Xi^2(u) \ne 0$ and it will satisfy $\mathfrak{H}_V(\Xi) = 0$,
in the notation of~\eqref{eq_xi_flat}.
\end{lem}

In the spatially flat ($\alpha = 0$) case, the equation $\mathfrak{H}_V(\Xi) = 0$ is sometimes known as the \emph{Hamilton-Jacobi equation} of
single field inflation~\cite{salopek-bond,muslimov}. The more general
system $\mathfrak{G}_V(\Pi,\Xi) = 0$ needed in the generic case ($\alpha
\ne 0$) does not seem to have been considered before. In the cosmology
literature, in the case of non-zero $\alpha$, an alternative system of
equations has been used~\cite{skenderis-townsend}, though one less
convenient for our purposes. There, a complex scalar field $Z(\phi)$ is
introduced, and plays the role of a ``super-potential'' (in the sense of
super-symmetry) for a ``pseudo-Killing'' spinor. The isometry class of
$(\alpha,f,\phi)$ determines the integrability conditions for $Z(\phi)$,
an algebraic relation between $\phi$, $Z(\phi)$ and $Z'(\phi)$.

\section{Geometric characterization} \label{sec_ideal}

In this section, we leverage the information from
Section~\ref{sec_FLRW_geom} to give necessary and sufficient conditions
to belong to the local isometry class of a regular FLRW or inflationary
spacetime, eventually proving our main Theorems~\ref{thm_FLRW_class}
and~\ref{thm_infl_class}.

The resulting systems of conditions will be of the IDEAL type, as
discussed in the Introduction, consisting of a list $\{ T_a[g,\phi] = 0
\}$, $a=1, \ldots, N$, of tensor equations built covariantly out of a
metric $g$, a scalar field $\phi$, and their derivatives. Each set of
equations will consist of roughly three parts: for the GRW structure,
for the FLRW structure, and for the specific isometry subclass.

\subsection{Special cases}

The two cases of FLRW spacetimes whose local isometry classes need to be
characterized separately from the general pattern given in the sequel are
the constant curvature spacetimes (Lemmas~\ref{lem_flat},
\ref{lem_const_curv}) and Einstein static universes
(Lemma~\ref{lem_esu}).

\begin{prop} \label{prop_special_FLRW}
Consider a Lorentzian manifold $(M,g)$, $\dim M = m+1 \ge 2$.

\begin{enumerate}
\item
Given a fixed constant $K$, if $(M,g)$ everywhere satisfies
\begin{equation}
	R_{ijkh} - K \frac{1}{2} (g\odot g)_{ijkh} = 0 ,
\end{equation}
then it is locally isometric to any other spacetime satisfying the same
condition.

\item
Given a fixed constant $K$, if $m>1$ and $(M,g)$ everywhere satisfies
\begin{align}
	W_{ijkh} &= 0 , &
	R_i^j \left(R_{jk} - (m-1) K g_{jk}\right) &= 0 , \\
	\nabla_i R_{jk} &= 0 , &
	\mathcal{R} - m(m-1) K &= 0 ,
\end{align}
while the $1$-dimensional kernel of $R_i^j$ is timelike, it is locally
isometric to an Einstein static universe with spatial sectional
curvature $K$. The value $K=0$ coincides with the flat case, $R_{ijkh} =
0$.
\end{enumerate}
\end{prop}
\begin{proof}
(a) This is standard; see for instance Theorem~2.4.11 in~\cite{Wolf2011}.

(b) When $m=1$, spatial slices are always flat, hence it is impossible
to have $K\ne 0$ spatial sectional curvature. When $K=0$, we are back in
the flat case, characterized by $R_{ijkh} = 0$, a special case of part
(a). This is why we take $m>1$. Direct calculation
(cf.~\ref{sec_Riemann}) shows that the above equations hold when $(M,g)$ is
an Einstein static universe with spatial sectional curvature $K\ne 0$.

Conversely, assume that we only know about $(M,g)$ that the above
equations hold, with $K\ne 0$. The algebraic equations on the $R_i^j$
tangent space endomorphism guarantee that it is diagonalizable with
precisely two distinct eigenvalues, $0$ and $(m-1)K$, with the kernel
being $1$-dimensional. Since $R_{ij}$ is symmetric, the kernel can only
be either timelike or spacelike (not null),%
	\footnote{Suppose the $1$-dimensional kernel $N$ of $R_i^j$ is null.
	From its invariant factors and the symmetry of $R_{ij}$, we have the
	following splittings of invariant subspaces: $N^\perp = N \oplus S$
	and $S^\perp = N \oplus N'$, where $S$ is necessarily spacelike,
	meaning that $N'$ is $1$-dimensional and has a non-zero eigenvalue.
	But, by the well-known Segre
	classification~\cite[\textsection5.1]{stephani-sols}, on $S^\perp$,
	$R_i^j$ can either have only a single degenerate eigenvalue or no null
	eigenvectors.} %
with the hypotheses constraining it to be timelike. Since $R_{ij}$ is
also covariantly constant, so is any unit vector $U^i$ in its kernel.
That is, $R_{ij} \nabla_X U^j = \nabla_X (R_{ij} U^j) = 0$ for any
$X^i$, which implies that $\nabla_X U^j = A_X U^j$ and $A_X = -U_j
\nabla_X U^j = -\frac{1}{2} \nabla_X (U_j U^j) = 0$. This gives us the
desired $\nabla_i U^j = 0$ conclusion.

The existence of a covariantly constant unit vector $U^i$ implies that
for any $x\in M$ and contractible open neighborhood $O \ni x$, the
holonomy action of $(O,g|_{O})$ at $x$ leaves invariant the subspace
spanned by $U^i$ at $x$ as well as its orthogonal complement (simply
note that contraction with $U^i$ commutes with parallel transport).
Under these conditions (Proposition~IV.5.2 in~\cite{kn1}), it is
possible to locally factor $(O,g|_{O})$ into a direct product of a
$1$-dimensional and an $m$-dimensional pseudo-Riemannian manifold,
$(I,-dt^2) \times (F,g^F)$, with $g^F$ of Riemannian signature.
Furthermore, the algebraic conditions on $W_{ijkh}$ and $R_{ij}$ imply
that $W^F_{ijkh} = 0$ and $R^F_{ij} = (m-1) K g^F_{ij}$, which means
that the spatial factor $(F,g^F)$ is locally of constant curvature with
sectional curvature $K$. In other words, we can locally describe $(M,g)$
as an FLRW spacetime with $\alpha = K$ and $f(t)=1$, which belongs
precisely to the desired Einstein static universe class.
\end{proof}

\subsection{FLRW spacetimes}

An FLRW spacetime (Definition~\ref{def_FLRW}) is a GRW spacetime
(Definition~\ref{def_GRW}) whose spatial slices have constant curvature
(Equation~\eqref{const_curv}). GRW spacetimes have been geometrically
characterized in two different but related ways by the existence of a
spatially conformal vector field $U$ by S\'anchez~\cite{Sanchez1998} and
of a concircular vector field $v$~by Chen~\cite{Chen2014}. Given Chen's
vector field $v$, the vector field $U = v / \sqrt{-v^2}$ satisfies the
conditions of S\'anchez. A recent survey of these and related geometric
characterization results of GRW spacetimes can be found
in~\cite{Mantica2017}.

Chen's condition is somewhat simpler, but we will only be able to make
use of it to characterize spatially curved, but not spatially flat FLRW
spacetimes. In one case it will be possible to produce Chen's vector
field $v$ directly from the spacetime curvature, in the other not.
S\'anchez's conditions work equally well also in the spatially flat case.
So, motivated by providing the simplest set of equations when possible,
we present both characterizations.

\begin{prop}[S\'anchez's conditions]\label{thm_Sanchez}
Let $(M,g)$ be a Lorentzian manifold, $\dim M = m+1 \ge 2$. It is
locally GRW at $x\in M$ if and only if there exists, on a neighborhood
of $x$, a unit timelike vector field $U$ that satisfies the conditions
\begin{align} 
	\mathfrak{P}_{jk}
	&:= U_{[ j} \nabla_{k]} \frac{\nabla^i U_i}{m} = 0 , \label{Uequations5} \\
	\mathfrak{D}_{ij}
	&:= \nabla_i U_j - \dfrac{\nabla_k U^k}{m} (g_{ij}+U_iU_j)=0 . \label{Uequations3}
\end{align}

\end{prop}
\begin{proof}
In one direction, given an FLRW metric in the form~\eqref{GRW_metric},
direct calculation shows that the above conditions are satisfied with
$U^i = (\del_t)^i$.

In the other direction, S\'anchez's Theorem~2.1 from~\cite{Sanchez1998}
shows that locally $(M,g)$ can be put into the form~\eqref{GRW_metric},
with $U^i = (\del_t)^i$. S\'anchez's original conditions look more
complicated, but they follow from ours by easy algebraic manipulations.
S\'anchez's hypotheses also include connectedness and simple
connectedness. But, from the proof, these can all be dropped for the
local result that we want.
\end{proof}

We have based the above result on the characterization of GRW spacetimes
that S\'anchez obtained independently~\cite[Theorem 2.1]{Sanchez1998} in
the process of a detailed investigation of the geometry of GRW
spacetimes. However, this characterization (existence of a
\emph{shear}-free, $\mathfrak{D}_{(ij)} = 0$, and \emph{twist}-free,
$\mathfrak{D}_{[ij]} = 0$, vector field $U$, with \emph{expansion} $\xi$
constant in directions orthogonal to $U$, $\mathfrak{P}_{ij} = 0$), at
least when applied to FLRW spacetimes, has been known already as far
back as~\cite[Theorem 2.5.1]{Ehlers1961,Ehlers1961en}, and has been referenced for
instance in~\cite[Section III.B]{EllisBruni}, \cite[Section
5.1]{Ellis1971}. Another independent source for these conditions seems
to be the unpublished thesis~\cite{Easley1991}, which has been
referenced in at least~\cite[p.124]{bee-globlor}.

\begin{prop}[Chen's conditions] \label{thm_GRW_concircular}
Consider a Lorentzian manifold $(M,g)$, $\dim M = m+1 \ge 2$. It is
locally GRW at $x\in M$ if and only if there exists, on a neighborhood
of $x$, a timelike vector field $v$ and a scalar $\mu$ that satisfy the
condition
\begin{equation} \label{concircular_vector}
	\nabla_i v_j = \mu g_{ij} .
\end{equation}
\end{prop}
A vector field satisfying~\eqref{concircular_vector} is called
\emph{concircular}.
\begin{proof}
In one direction, given GRW metric in the form~\eqref{GRW_metric},
direct calculation shows that we can take $v^i = f (\del_t)^i$ and $\mu
= f'$.

Chen's Theorem~1 from~\cite{Chen2014} shows that locally $(M,g)$ can be
put into the form~\eqref{GRW_metric}, with $v^i = f (\del_t)^i$. Chen
stated this result for $m+1 \ge 3$. However, the same proof also works
when $m+1=2$. It is easiest to see by showing that the concircular
condition~\eqref{concircular_vector} implies that $U^i =
v^i/\sqrt{-v^2}$ satisfies S\'anchez's conditions, independently of the
dimension. Let $\phi = \sqrt{-v^2}$, so that $v^i = \phi U^i$. From the
$U^j \nabla_i U_j = 0$ identity, the concircular condition decomposes
into
\begin{multline}
	U_j \nabla_i \phi + \phi \nabla_i U_j
		= -\mu U_i U_j + \mu (g_{ij} + U_i U_j)
	\\ \iff
	\nabla_i \phi = -\mu U_i , \quad
	\nabla_i U_j = \frac{\mu}{\phi} (g_{ij} + U_i U_j) .
\end{multline}
Then $U_{[i} U_{j]} = 0$ implies $U_{[i} \nabla_{j]} \phi = 0$, and
$\nabla_{[i} \nabla_{j]} \phi = 0$ implies $U_{[i} \nabla_{j]} \mu = 0$.
Finally, noting that $\mu = U^i \nabla_i \phi = \frac{\phi}{m} \nabla^i
U_i$ and eliminating both $\phi$ and $\mu$ gives us S\'anchez's conditions
$\mathfrak{P}_{ij} = 0$ and $\mathfrak{D}_{ij} = 0$.
\end{proof}

The concircular condition can be rewritten slightly for our convenience.
\begin{lem} \label{lem_concircular}
Let $U$ be a vector field, $\nu$ and $\phi$ smooth functions, with $\phi
> 0$, and $k$ a constant. Then the condition
\begin{equation} \label{alt_concircular}
	\nabla_i U_j + k \frac{\nabla_i \phi}{\phi} U_j = \nu g_{ij}
\end{equation}
implies that $v = \phi^k U$ is a concircular vector field. In
particular, $U_{[i} \nabla_{j]} \phi = 0$.
\end{lem}
\begin{proof}
The concircular condition with $v = \phi^k U$ and $\mu = \phi^k \nu$ is
equivalent to $\phi^{-k} \nabla_i (\phi^k U_j) = \phi^{-k} \mu g_{ij}$,
which when expanded gives precisely Equation~\eqref{alt_concircular}. In
GRW form~\eqref{GRW_metric}, $\phi^k U^i = f (\del_t)^i$ and $\nabla_j f
= -f' U_j$, from which follows the desired condition on $\nabla_j \phi$.
\end{proof}

\begin{prop} \label{prop_const_curv}
Consider a GRW spacetime $(M,g) \cong (I\times F , -dt^2 + f^2 g^F)$,
$\dim M = m+1 \ge 2$. Set $U^i = (\del_t)^i$ and recall the notation of
Definition~\ref{curvature_cond}.

The $(F,g^F)$ factor is locally of constant curvature if and only if the
CCD tensor (see Definition~\ref{curvature_cond}) vanishes and the spatial scalar curvature is constant,
\begin{equation} \label{ccc} 
	\mathfrak{C}_{ijkh} = 0
	\quad \text{and} \quad
	U_{[i} \nabla_{j]} \zeta = 0 .
\end{equation}
If in addition the spatial scalar curvature or
equivalently the ZCD tensor (see Definition~\ref{curvature_cond}) also vanishes, $\zeta = 0$ or
\begin{equation} \label{zcc} 
	\mathfrak{Z}_{ijkh} = 0 ,
\end{equation}
then $(F,g^F)$ is actually flat.
\end{prop}
\begin{proof}
From Equation~\eqref{GRW_CCDT}, $\mathfrak{C}_{ijkh} = 0$ is equivalent
to
\begin{equation}
	R^F_{ijkh}
	= \frac{1}{m(m-1)} \frac{\mathcal{R}^F}{2}
			(g^F\odot g^F)_{ijkh} .
\end{equation}
while $U_{[i} \nabla_{j]} \zeta = 0$ and~\eqref{GRW_ZCDT} imply that
$\mathcal{R}^F$ is a constant. Hence, $(F,g^F)$ is of constant
curvature. Furthermore, either of the conditions $\zeta = 0$ or
$\mathfrak{Z}_{ijkh} = 0$ implies that $R^F_{ijkh} = 0$ and hence that
$(F,g^F)$ is flat.
\end{proof}

\subsection{FLRW local isometry classes} \label{FLRW_isom_class}

Within the class of FLRW spacetimes, two metrics in the
form~\eqref{GRW_metric} with different $(\alpha,f)$ parameters may or
may not be isometric. Below, we give the results that allow us to
classify FLRW metrics into isometry classes.

The obvious form-preserving transformations, time translation,
reflection and rescaling, relate any FLRW metric to a 2-parameter family
of (locally) isometric metrics. We state this result directly for FLRW
spacetimes with scalar, which will come in useful later in
Section~\ref{infl_isom_class}. As mentioned in the introduction, we can
reduce to the case of no scalar field by setting the scalar field to
zero.
\begin{prop} \label{prop_isometries}
Consider two inflationary spacetimes $(M_i,g_i,\phi_i)$, $i=1,2$, with
corresponding spatial sectional curvature, warping function and scalar
field triples $(\alpha_i, f_i, \phi_i)$, $i=1,2$. If for every $x\in
M_1$ with $t_1=t(x)$ in the domain of $(f_1,\phi_1)$ there exists an
open interval $(t_1-\delta,t_1+\delta)$ still in the domain of
$(f_1,\phi_1)$, with $\delta>0$, and an interval $(t_2-\delta,
t_2+\delta)$ in the domain of $(f_2,\phi_2)$ such that
\begin{equation} \label{isom_structure}
	\left\{\begin{aligned}
		\alpha_1 &= A^2 \alpha_2 , \\
		f_1(t) &= A f_2(st+t_0) , \\
		\phi_1(t) &= \phi_2(st+t_0) ,
	\end{aligned}\right.
\end{equation}
for some constants $s\in \{+1,-1\}$, $A \ne 0$ and every $t\in
(t_1-\delta,t_1+\delta)$, then $(M_1,g_1,\phi_1)$ is locally isometric
to $(M_2,g_2,\phi_2)$ at $x\in M_1$.
\end{prop}
\begin{proof}
The result follows from noting that an FLRW metric in standard form
$-dt^2 + f(t)^2 \tilde{g}^F$ is locally isometric to each of $-dt^2 +
f(-t)^2 \tilde{g}^F$, $-dt^2 + f(t+t_0)^2 \tilde{g}^F$ and to $-dt^2 +
(Af(t))^2 (\tilde{g}^F/A^2)$.
\end{proof}

We will now show that, under certain conditions, two FLRW metrics with
parameters $(\alpha_1,f_1)$ and $(\alpha_2,f_2)$ are locally isometric
\emph{if and only if} they belong to the same 2-parameter family as in
Proposition~\ref{prop_isometries}. To describe such a 2-parameter family
of $(\alpha,f)$ intrinsically, we will look for a differential equation
satisfied by every element of that family and only elements of that
family. Heuristically, we should look for either a second order equation
for $f$ or a first order equation for $f$ depending also on the
parameter $\alpha$, either of which will generically have a 2-parameter
general solution.

The following helpful lemma follows easily from standard ODE existence
and uniqueness theory~\cite{arnold}.

\begin{lem} \label{lem_isom_if}
Consider a smooth real function $G$ defined on an open interval $J$,
two nonzero real constants $\alpha_1$ and $\alpha_2$, and two nowhere
vanishing smooth real functions $f_1(t)$ and $f_2(t)$ defined
respectively on the open intervals $I_1$ and $I_2$.

\begin{enumerate}
\item
Suppose $G>0$ and that the pairs $(\alpha_1,f_1)$ and
$(\alpha_2,f_2)$ both satisfy the differential equation
\begin{equation}
	(f'/f)^2 = G(\alpha/f^2)
\end{equation}
and that there exist $t_1 \in I_1$ and $t_2 \in I_2$ such that
$\frac{\alpha_1}{f_1(t_1)^2} = \frac{\alpha_2}{f_2(t_2)^2} \in J$. Then
there exist constants $s\in \{+1,-1\}$, $t_0$, $A\ne 0$ and $\delta > 0$
such that $t_2 = st_1+t_0$, as well as
\begin{equation}
	\alpha_1 = A^2 \alpha_2
	\quad \text{and} \quad
	f_1(t) = A f_2(st+t_0)
\end{equation}
for every $t \in (t_1-\delta,t_1+\delta)$.

\item
Suppose that the functions $f_1$ and $f_2$ both satisfy the
differential equation
\begin{equation}
	f''/f = G\left((f'/f)^2\right)
\end{equation}
and that there exist $t_1 \in I_1$ and $t_2 \in I_2$ such that
$\frac{f'_1(t_1)}{f_1(t_1)} = \frac{f'_2(t_2)}{f_2(t_2)} \in J$. Then
there exist constants $s\in \{+1,-1\}$, $t_0$, $A\ne 0$ and $\delta > 0$
such that $t_2 = st_1+t_0$, as well as
\begin{equation}
	f_1(t) = A f_2(st+t_0)
\end{equation}
for every $t \in (t_1-\delta,t_1+\delta)$.
\end{enumerate}
\end{lem}

We are finally in a position to define and classify all regular FLRW
spacetimes into families and to describe the parameters needed identify
an isometry class within each family.

\begin{lem} \label{lem_reg_FLRW}
Two regular FLRW spacetimes (those belonging to one of the families
identified in Definition~\ref{def_reg_FLRW}) are isometric to each other
(Definition~\ref{def_loc_isom}) if and only if they belong to the same
parametrized family and the corresponding parameters are identical.
\end{lem}

\begin{proof}
Let us fix $m$, noting that two isometric spacetimes must have the same
dimension. To show that two spacetimes cannot be isometric, it is
sufficient to point out an identity or inequality that is satisfied by
curvature scalars or tensors on one spacetime but not on the other. With
that in mind, recall (in the notation of Theorem~\ref{thm_FLRW_class})
that for FLRW spacetimes, $\xi = f'/f$, $\bm{\eta} = f''/f - f'^2/f^2$
and $\zeta = \alpha/f^2$, which are all curvature scalars as long as
they are defined with respect to a vector field $U$ that is also defined
from pure, such as the choices $U = U_\mathcal{R}$ or $U_\mathcal{B}$.
To show that all the representatives of a family with identical
parameters are all isometric to each other, there will be two
possibilities to consider. Either the representative is unique, which is
the trivial case. Or, all representatives are selected by satisfying a
differential equation. By invoking Lemma~\ref{lem_isom_if}, we can be
sure that two solutions to such an equation (with all parameters fixed),
if they can be matched up at at least one point, are in fact locally
isometric around that point. If the domains of these solutions can also
be matched up, then it is clear that they are also globally isometric.

(a) For each $K$, there is a unique representative in $\CC^m_K$.
The scalar curvature $\mathcal{R} = m(m+1)K$ distinguishes the different
values of $K$.

(b) Again, for each $K\ne 0$, there is a unique representative in
$\ESU^m_K$. The scalar curvature $\mathcal{R} = m(m-1)K$ distinguishes
the different values of $K$. Comparing the formulas from
Section~\ref{sec_Riemann} and Proposition~\ref{prop_special_FLRW}(b),
the structure of the Ricci tensor $R_{ij}$ distinguish $\ESU^m_K$ from
any spacetime of constant curvature.

(c) The representatives of $\CSC^{m,0}_{K,J}$ satisfy an
equation like in Lemma~\ref{lem_isom_if}(b). The scalar curvature
$\mathcal{R} = m(m+1)K$ distinguishes the different values of $K$, and
setting $U = U_\mathcal{B}$ the range $J = \xi^2(I)$ distinguishes the
different intervals $J$. Also, from Lemma~\ref{lem_csc}, $(\nabla
\mathcal{B})^2 < 0$ distinguishes these spacetimes from those of parts
(a) and (b), where $\mathcal{B}'=0$.

(d) The representatives of the class $\CSC^m_{K,\Omega,J}$ satisfy an
equation like in Lemma~\ref{lem_isom_if}(a). The scalar curvature
$\mathcal{R} = m(m+1)K$ distinguishes the different values of $K$, and
setting $U = U_{\mathcal{B}}$, the constant $\kappa\Omega =
(\xi^2+\zeta-K) / |\zeta|^{\frac{m+1}{2}}$ (Lemma~\ref{lem_csc}) and
range $J = \zeta(I)$ distinguishes the different values of $\Omega$ and
$J$. Again, $(\nabla \mathcal{B})^2 < 0$ distinguishes these spacetimes
from those of parts (a) and (b), while $\zeta \ne 0$ distinguishes them
from those of part (c) where $\zeta = 0$.

(e) The representatives of the class $\FLRW^{m,0}_{P,J}$ satisfy an
equation like in Lemma~\ref{lem_isom_if}(b). Setting $U =
U_\mathcal{R}$, the identity $\bm{\eta} + \frac{m}{2}\xi^2 = -\kappa P(\xi^2)$ and the
range $J = \xi^2(I)$ distinguish different values of the $P$ and $J$
parameters. Also, combining the constraints on $P$ and
Lemma~\ref{lem_flat_FLRW}, $(\nabla \mathcal{R})^2 < 0$ distinguishes
these spacetimes from those of parts (a), (b), (c) and (d), where
$\mathcal{R}' = 0$.

(f) The representatives of the class $\FLRW^m_{E,J}$ satisfy an equation
like in Lemma~\ref{lem_isom_if}(b). Setting $U = U_\mathcal{R}$, the
identity $\xi^2+\zeta = \kappa E(\zeta)$ and the range $J = \zeta(I)$
distinguish different values of the $E$ and $J$ parameters. Again,
combining the constraints on $E$ and Lemma~\ref{lem_gen_FLRW}, $(\nabla
\mathcal{R})^2 < 0$ distinguishes these spacetimes from those of parts
(a), (b), (c) and (d), while $\zeta \ne 0$ distinguishes them from those
of part (e), where $\zeta = 0$.
\end{proof}

We are now finally in a position to prove our main result about IDEAL
characterizations of regular FLRW spacetimes.
\begin{proof}[Proof of Theorem~\ref{thm_FLRW_class}]
The goal is to prove that, for each of the cases listed in
Table~\ref{tab_FLRW_class}, a spacetime satisfies the listed equations
(and inequalities) if and only if it is locally isometric
(Definition~\ref{def_loc_isom}) to one of the regular FLRW spacetimes
listed in Definition~\ref{def_reg_FLRW}. In one direction (a regular
FLRW spacetime satisfies the corresponding conditions), this is
essentially the content of Lemma~\ref{lem_reg_FLRW}. It remains to show
the converse.

(a) The constant curvature case is standard
(Proposition~\ref{prop_special_FLRW}(a)).

(b) We have already proven the desired conclusion in the Einstein static
universe case in Proposition~\ref{prop_special_FLRW}(b).

(c,e) With the appropriate definition of the unit timelike vector field
$U$, according to Proposition~\ref{thm_Sanchez}, the equations
$\mathfrak{P}_{ij} = 0$ and $\mathfrak{D}_{ij} = 0$ are sufficient to
locally put the spacetime in GRW form~\eqref{GRW_metric}, while
according to Proposition~\ref{prop_const_curv} the equation
$\mathfrak{Z}_{ijkh} = 0$ implies that the spatial slices are flat and
hence the spacetime is locally FLRW. The remaining conditions place the
spacetime in the unique corresponding local regular FLRW isometry class,
as per Lemma~\ref{lem_reg_FLRW}(c,e).

(d,f) With the appropriate definition of the unit timelike vector field
$U$, according to Proposition~\ref{thm_GRW_concircular} and
Lemma~\ref{lem_concircular}, the equation $\nabla_i U_j-\frac{\nabla_i
\zeta}{2\zeta} U_j - \xi g_{ij} = 0$ is sufficient to locally put the
spacetime in GRW form~\eqref{GRW_metric} and show that $\zeta$ is
constant along the spatial slices, while according to
Proposition~\ref{prop_const_curv} the additional equation
$\mathfrak{C}_{ijkh} = 0$ implies that the spatial slices are of
constant curvature and hence the spacetime is locally FLRW. The
remaining conditions place the spacetime in the unique corresponding
local regular FLRW isometry class, as per Lemma~\ref{lem_reg_FLRW}(c,e).
\end{proof}

\subsection{Inflationary local isometry classes} \label{infl_isom_class}

Within the class of inflationary spacetimes $(M,g,\phi)$, two spacetimes
in the form~\eqref{GRW_metric} and with $\phi = \phi(t)$, with different
$(\alpha,f,\phi)$ parameters may or may not be isometric. Below, we give
the results that allow us to classify inflationary spacetimes into
isometry classes (Definition~\ref{def_loc_isom}).

Recall that Proposition~\ref{prop_isometries} gives a sufficient
condition for local isometry. We will now show that, under certain
conditions, two inflationary spacetimes with parameters
$(\alpha_i,f_i,\phi_i)$, $i=1,2$, are locally isometric \emph{if and
only if} they belong to the same $2$-parameter family as in
Proposition~\ref{prop_isometries}. As in Section~\ref{FLRW_isom_class},
we will look for an ODE system, jointly satisfied by any locally
isometric $(\alpha,f,\phi)$ triples, with a 2-parameter general
solution. The following helpful lemma, the analog of
Lemma~\ref{lem_isom_if}, again follows easily from standard ODE
existence and uniqueness theory~\cite{arnold}.

\begin{lem} \label{lem_isom_if_scalar}
Consider a smooth real function $V\colon J \to \mathbb{R}$ defined on an
open interval, two non-zero real constants $\alpha_i$, $i=1,2$, and two
pairs of smooth real functions $(f_i,\phi_i)$ defined on intervals
$I_i$, $i=1,2$, with either $f_i$ nowhere vanishing.

\begin{enumerate}
\item
Suppose that $\Pi,\Xi\colon J \to \mathbb{R}$ are smooth real functions
that satisfy the $\mathfrak{G}_V(\Pi,\Xi) = 0$, in the notation
of~\eqref{eq_pi_xi}. Suppose also that the triples
$(\alpha_i,f_i,\phi_i)$, $i=1,2$, both satisfy the system of
differential equations
\begin{equation}
\begin{aligned}
	\phi' &= \Pi(\phi) , \\
	\frac{f'}{f} &= \Xi(\phi) , \\
	\frac{\alpha}{f^2} &= \kappa \frac{\Pi^2(\phi) + V(\phi)}{m(m-1)}
		- \Xi^2(\phi) ,
\end{aligned}
\end{equation}
and that there exist $t_i \in I_i$, $i=1,2$, such that $\phi_1(t_1) =
\phi_2(t_2) \in J$ and $\frac{\alpha_1}{f_1(t_1)^2} =
\frac{\alpha_2}{f_2(t_2)^2}$. Then there exist constants $t_0$, $A\ne 0$
and $\delta > 0$ such that
\begin{equation}
	\alpha_1 = A^2 \alpha_2 , \quad
	f_1(t) = A f_2(t+t_0)
	\quad \text{and} \quad
	\phi_1(t) = \phi_2(t+t_0)
\end{equation}
for every $t\in (t_1-\delta,t_1+\delta)$.

\item
Suppose that $\Xi\colon J \to \mathbb{R}$ is a smooth real function.
Suppose also that the pairs $(f_i,\phi_i)$, $i=1,2$, both satisfy the
system of differential equations
\begin{equation}
\begin{aligned}
	\phi' &= -\frac{(m-1)}{\kappa} \del_\phi \Xi(\phi) , \\
	\frac{f'}{f} &= \Xi(\phi) ,
\end{aligned}
\end{equation}
and that there exist $t_i \in I_i$, $i=1,2$, such that $\phi_1(t_1) =
\phi_2(t_2) \in J$. Then there exist constants $t_0$, $A\ne 0$ and
$\delta > 0$ such that
\begin{equation}
	f_1(t) = A f_2(t+t_0)
	\quad \text{and} \quad
	\phi_1(t) = \phi_2(t+t_0)
\end{equation}
for every $t \in (t_1-\delta,t_1+\delta)$.
\end{enumerate}
\end{lem}

We are finally in a position to define and classify all regular
inflationary spacetimes into families and to describe the parameters
needed to identify an isometry class within each family.

\begin{lem} \label{lem_reg_infl}
Two regular inflationary spacetimes (those belonging to one of the
families identified in Definition~\ref{def_reg_infl}) are isometric to
each other (Definition~\ref{def_loc_isom}) if and only if they belong to
the same parametrized family and the corresponding parameters are
identical.
\end{lem}

The following proofs are very much analogous to the proofs of
Lemma~\ref{lem_reg_FLRW} and Theorem~\ref{thm_FLRW_class}, but we will
write them in a mostly self-contained way.

\begin{proof}[Proof of Lemma~\ref{lem_reg_infl}]
Let us fix $m$, noting that two isometric spacetimes must have the same
dimension. To show that two spacetimes with scalar cannot be isometric,
it is sufficient to point out an identity or inequality that is
satisfied by curvature scalars or tensors, possibly together also with
scalars or tensors covariantly obtained from the scalar field, on one
spacetime but not on the other. With that in mind, recall (in the
notation of Theorems~\ref{thm_FLRW_class} and~\ref{thm_infl_class}), that
for inflationary spacetimes $\xi = f'/f$, $\bm{\eta} = f''/f - f'^2/f^2$
and $\zeta = \alpha/f^2$, which are all curvature scalars, as long as
they are defined with respect to a vector field $U$ that is also defined
from either pure curvature or from the scalar field, such as the choices
$U = U_\mathcal{R}$, $U_\mathcal{B}$ or $U_\phi$. To show that all the
representatives of a family with identical parameters are all isometric
to each other, there will be two possibilities to consider. Either the
representative is unique, which is the trivial case. Or, all
representatives are selected by satisfying a differential equation. By
invoking Lemmas~\ref{lem_isom_if_scalar} or~\ref{lem_isom_if}, we can be
sure that two solutions to such an equation (with all parameters fixed),
if they can be matched up at at least one point, are in fact locally
isometric around that point. If the domains of these solutions can also
be matched up, then it is clear that they are also globally isometric.

(a) For each $\Lambda$ (hence $K = \frac{2}{m(m-1)} \Lambda$) and
$\Phi$, there is a unique representative in $\CC^m_K \CS_\Phi$. The
scalar curvature $\mathcal{R} = m(m+1) K$ and the scalar field $\phi =
\Phi$ distinguish the different values of these parameters.

(b) For each $\rho>0$ (hence $K=\frac{2}{m(m-1)} \kappa \rho$) and
interval $J\subseteq \mathbb{R}$, there is a unique representative in
$\ESU^m_K \CES_{\rho,J}$. The scalar curvature $\mathcal{R} = m(m-1) K$
and the range $J = \phi(I)$ distinguish different values $\rho$ and $J$.
The condition $(\nabla\phi)^2 < 0$ distinguishes these spacetimes from
those in part (a), where $\phi' = 0$.

(c) The representatives of $\MMS^{m,0}_{\Lambda,J,J'}$ satisfy the
equations $f''/f + (m-1)f'^2/f^2 = \frac{2\Lambda}{(m-1)}$ which is like
in Lemma~\ref{lem_isom_if}(a), and
\begin{equation}
	\phi' = -\sqrt{\frac{1}{\kappa}
		\left(\frac{f'^2}{f^2}-\frac{2\Lambda}{m(m-1)}\right)} ,
\end{equation}
since by hypothesis $\phi' < 0$. Thus, the first equation shows that the
underlying Lorentzian spacetimes are isometric for identical $\Lambda$
and $J'$. The second equation shows, by applying once again standard ODE
existence and uniqueness theory, that the inflationary spacetimes are
also isometric (as spacetimes with scalar) for identical $J$. With the
choice $U = U_\phi$, the curvature scalars $\bm{\eta} + m\xi^2 =
\frac{2\Lambda}{(m-1)}$ and the ranges of $J = \phi(I)$, $J' = \xi(I)$
distinguishes different $\Lambda$, $J$ and $J'$. The implication that
$\xi=f'/f\ne 0$ and $\phi' \ne 0$ distinguish these spacetimes from
those of parts (a) and (b).

(d) The representatives of the class $\MMS^m_{\Lambda,\Omega,J}$ satisfy
an equation like in Lemma~\ref{lem_isom_if_scalar}(a), namely
\begin{equation}
	\phi' = -\sqrt{\Omega} \frac{|\alpha|^{\frac{m}{2}}}{f^m} ,
	\quad
	\frac{f'}{f} = \pm \sqrt{\frac{2\Lambda + \kappa \Omega
	|\alpha|^m/f^{2m}}{m(m-1)} - \frac{\alpha}{f^2}} ,
\end{equation}
where the $\pm$ sign is determined by whether $0 < J'$ or $J' < 0$.
With the choice $U = U_\phi$, the curvature scalars $\bm{\eta} + m\xi^2
= \frac{2\Lambda}{(m-1)}$, $|\zeta|^{-m} (\xi^2+\zeta) = \frac{\kappa
\Omega}{m(m-1)}$ and the range $J = \phi(I)$ distinguish different
$\Lambda$, $\Omega$ and $J$. The implication that $\xi=f'/f\ne 0$ and
$\phi' \ne 0$ distinguish these spacetimes from those of parts (a) and
(b), while $\zeta = \kappa\frac{\Pi^2(\phi)+V(\phi)}{m(m-1)}-\Xi^2(\phi)
\ne 0$ distinguishes them from those of part (c), where $\zeta = 0$.

(e) The representatives of the class $\NKG^{m,0}_{V,\Xi,J}$ satisfy an
equation like in Lemma~\ref{lem_isom_if_scalar}(b). With the choice $U =
U_\phi$, the identities $\phi' = -\frac{(m-1)}{\kappa} \del_\phi
\Xi(\phi)$, $\xi = \Xi(\phi)$ and the range $J = \phi(I)$ distinguish
different $\Xi$, and $J$. It is important to note that for any solution
of $\mathfrak{H}_V(\Xi)$, $-\Xi$ is also a solution that defines another
spacetime isometric to a given one via $t \mapsto -(t-t_0)$ for some
$t_0$. We have broken this degeneracy by the $\frac{1}{\kappa} \del_u
\Xi(u) > 0$ requirement (due to using $U_\phi$ and not $-U_\phi$), so
distinct $\Xi$ imply non-isometric spacetimes. The identity $\bm{\eta} +
m\xi^2 = \kappa\frac{V(\phi)}{(m-1)}$, with non-constant $V(\phi)$,
distinguishes these spacetimes from those in parts (a), (b), (c) and
(d), where the left-hand-side would have been constant.

(f) The representatives of the class $\NKG^m_{V,\Pi,\Xi,J}$ satisfy an
equation like in Lemma~\ref{lem_isom_if_scalar}(a). With the choice $U =
U_\phi$, the identities $\phi' = \Pi(\phi)$, $\xi = \Xi(\phi)$ and range
$J = \phi(I)$ distinguish different $\Pi$, $\Xi$ and $J$. It is
important to note that for any solution of $\mathfrak{G}_V(\Pi,\Xi)$,
$(-\Pi,-\Xi)$ is also a solution that defines another spacetime
isometric to a given one via $t \mapsto -(t-t_0)$ for some $t_0$. We
have broken this degeneracy by the $\Pi < 0$ requirement (due to using
$U_\phi$ and not $-U_\phi$), so distinct $\Xi$ imply non-isometric
spacetimes. The identity $\bm{\eta} + m\xi^2 =
\kappa\frac{V(\phi)}{(m-1)}$, with non-constant $V(\phi)$, distinguishes
these spacetimes from those in parts (a), (b), (c), and (d), where the
left-hand-side would have been constant, while $\zeta \ne 0$
distinguishes them from those of part (e), where $\zeta = 0$.
\end{proof}

We are now finally in a position to prove our main result about IDEAL
characterizations of regular inflationary spacetimes.

\begin{proof}[Proof of Theorem~\ref{thm_infl_class}]
The goal is to prove that, for each of the cases listed in
Table~\ref{tab_infl_class}, a spacetime satisfies the listed equations
(and inequalities) if and only if it is locally isometric
(Definition~\ref{def_loc_isom}) to one of the regular inflationary
spacetimes listed in Definition~\ref{def_reg_infl}. In one direction (a
regular inflationary spacetime satisfies the corresponding condition),
this is essentially the content of Lemma~\ref{lem_reg_infl}. It remains
to show the converse.

(a) When $\phi = \Phi$ is a constant, so is $V(\phi) = \frac{2}{\kappa}
\Lambda$, which we have parametrized for our convenience with $\Lambda$.
Then the Einstein-Klein-Gordon equations become the cosmological vacuum
equations $R_{ij} - \frac{1}{2}\mathcal{R} g_{ij} + \Lambda g_{ij} = 0$,
which under the FLRW hypotheses have only the constant curvature
solution.

(b) The existence of a timelike covariantly constant vector $U =
U_\phi$, $\nabla_i U_j = 0$, implies that the spacetime decomposes into
a direct sum, with one of the factors being of constant curvature, since
the CCD tensor $\mathfrak{C}_{ijkh} = 0$ (see Definition~\ref{curvature_cond}) vanishes and the spatial scalar
curvature $\zeta = \frac{2\kappa}{m(m-1)} \rho$ is constant
(Proposition~\ref{prop_const_curv}); see the proof of
Proposition~\ref{prop_special_FLRW}(b) for details. The conclusion, as
desired, is that the spacetime is an Einstein static universe and the
equation $\phi' = -\sqrt{2\rho/m}$ means that we can choose the time
coordinate to put $\phi(t)$ precisely into the form in
Lemma~\ref{lem_ces}.

(c,d) With the vector field $U = U_\phi$, according to
Proposition~\ref{thm_GRW_concircular} and Lemma~\ref{lem_concircular},
the equation $\nabla_i U_j - \frac{\nabla_i \phi'}{m\phi'} U_j - \xi
g_{ij} = 0$ is sufficient to locally put the spacetime into GRW
form~\eqref{GRW_metric} and show that $\phi'$ is constant along spatial
slices. In case (c), the vanishing of the ZCD tensor
$\mathfrak{Z}_{ijkh} = 0$ (see Definition~\ref{curvature_cond}) implies that the spatial slices are flat. In
case (d), the equation $\phi' = -\sqrt{\Omega} |\zeta|^{\frac{m}{2}}$
shows that $\zeta$ is also constant on spatial slices, and together with
the vanishing of the CCD tensor $\mathfrak{C}_{ijkh} = 0$ this implies
that the spatial slices are of constant curvature. In both cases we have
referred to Proposition~\ref{prop_const_curv}, and in both case we have
established that the spacetime is locally FLRW. Now, recalling the
identities $\xi = f'/f$, $\bm{\eta} = f''/f - f'^2/f^2$ and $\zeta =
\alpha/f^2$, the remaining conditions in each case clearly show that the
spacetime is locally isometric to the corresponding reference class in
Definition~\ref{def_reg_infl}(c) or (d).

(e,f) With the vector field $U = U_\phi$, according to
Proposition~\ref{thm_Sanchez}, the equations $\mathfrak{P}_{ij} = 0$ and
$\mathfrak{D}_{ij} = 0$ are sufficient to locally put the spacetime into
GRW form~\eqref{GRW_metric}. In case (e), the vanishing of the ZCD
tensor $\mathfrak{Z}_{ijkh} = 0$ implies that the spatial slices are
flat. In case (f), the equations $\phi' = \Pi(\phi)$, $\xi = \Xi(\phi)$
show that $\zeta = \kappa \frac{\phi'^2+V(\phi)}{m(m-1)} - \xi^2$ is
then constant along spatial slices (slices of constant $\phi$), and
together with the vanishing of the CCD tensor $\mathfrak{C}_{ijkh} = 0$
this implies that the spatial slices are of constant curvature. In both
cases we have referred to Proposition~\ref{prop_const_curv}, and in both
cases we have established that the spacetime is locally FLRW. Now,
recalling the identities $\xi=f'/f$, $\bm{\eta} = f''/f - f'^2/f^2$ and
$\zeta = \alpha/f^2$, the remaining conditions in each case clearly show
that the spacetime is locally isometric to the corresponding reference
class in Definition~\ref{def_reg_infl}(e) or (f).
\end{proof}

\paragraph{Acknowledgments.}
The authors thank M.~S\'anchez for some helpful discussions and
Y.~Urakawa for pointing out reference~\cite{skenderis-townsend}. Some of
the results presented here are based on the MSc thesis (Dept.\ of
Physics, Pavia, 2016) of GC~\cite{Canepa2016}. The work of CD was
supported by the University of Pavia. GC acknowledges partial support of
SNF Grant No. 200020\_172498/1. GC was also (partly) supported by the
NCCR SwissMAP, funded by the Swiss National Science Foundation, and by
the COST Action MP1405 QSPACE, supported by COST (European Cooperation
in Science and Technology). IK was partially supported by the ERC
Advanced Grant 669240 QUEST ``Quantum Algebraic Structures and Models''
at the University of Rome 2 (Tor Vergata). IK also thanks the University
of Z\"urich for its hospitality during the completion of part of this
work.

\printbibliography

\end{document}